\documentclass[10pt,journal,twocolumn,twoside]{IEEEtran}

\usepackage{cite}
\usepackage{amsmath,amssymb,amsfonts}
\allowdisplaybreaks[2]
\usepackage{graphicx}
\usepackage{hyperref,color,xcolor}
\usepackage{textcomp}
\usepackage{comment}
\usepackage{tabularx}
\usepackage{algpseudocode}
\usepackage{tikz}
\usetikzlibrary{positioning,arrows.meta,calc}

\usepackage{amssymb}
\usepackage{graphicx}
\usepackage{esint}
\usepackage{
epsfig,graphics,url,times,mathrsfs,array,
latexsym,fancyhdr,xspace,wrapfig,xcolor,multirow,
helvet,sidecap,bm
}
\usepackage{bbm}

\usepackage{amsthm}

\usepackage{mathrsfs}
\usepackage{algorithm,algpseudocode}
\usepackage{multirow}
\usepackage{diagbox}
\usepackage{booktabs}
\usepackage{multicol}
\usepackage[nocomma]{optidef}
\theoremstyle{remark}
\newtheorem{remark}{Remark}
\theoremstyle{plain}
\newtheorem{thm}{Theorem}
\theoremstyle{plain}

\theoremstyle{definition}
\newtheorem{defn}{Definition}
\newtheorem{lemma}{Lemma}
\theoremstyle{plain}

\newtheorem{corollary}{Corollary}
\newtheorem{condition}{Condition}

\newcommand{\KL}{\mathbb{D}}
\newcommand{\dP}{\mathrm{d}\mathbb{P}}
\newcommand{\dG}{\mathrm{d}\mathbb{G}}
\newcommand{\dQ}{\mathrm{d}\mathbb{Q}}
\newcommand{\rx}{\bold{x}}
\newcommand{\ry}{\bold{y}}
\newcommand{\ru}{\bold{u}}
\newcommand{\ra}{\bold{a}}
\newcommand{\rw}{\bold{w}}
\newcommand{\rc}{\bold{c}}
\newcommand{\rh}{\bold{h}}
\newcommand{\rv}{\bold{v}}
\newcommand{\rz}{\bold{z}}
\newcommand{\rg}{\bold{g}}
\newcommand{\T}{\top}
\newcommand{\SI}{\text{SI}}
\newcommand{\Tr}{\text{Tr}}
\DeclareMathOperator{\range}{range}
\DeclareMathOperator{\rank}{rank}

\begin{document}
\title{Rate Lower Bounds for Multi-Sensor Networked Control and Their Application to Decoder Side Information}
\author{Sijie Li, Takashi Tanaka, and Hyeji Kim
\thanks{Sijie Li and Hyeji Kim are with the University of Texas at Austin, Austin, TX 78712 USA (email: sijieli@utexas.edu; hyeji@utexas.edu). }
\thanks{Takashi Tanaka is with Purdue University, West Lafayette, IN 47907 USA (e-mail: tanaka16@purdue.edu).}
}

\maketitle

\begin{abstract}
This paper studies the tradeoff between communication rate and control performance for a linear quadratic Gaussian (LQG) plant with multiple sensors/encoders and one decoder/controller. Assuming that each encoder uses a prefix code, we derive a lower bound based on causally conditioned directed information and show that a class of linear policies suffices to optimize this weighted-sum lower bound under the full-observation condition, reducing the optimization to a finite-dimensional problem. We then apply our analysis to the point-to-point case with side information (SI) available only at the decoder, which can be viewed as a special case of the multi-encoder setting. We recover the structural result for the sequential rate-distortion function (SRDF) in this setting. We further present semidefinite programming (SDP) formulations for a broader class of system parameters, recovering existing formulations as special cases, and derive an equivalent single-letter SDP formulation for infinite-horizon time-invariant systems under certain conditions.

\end{abstract}

\section{Introduction}












This paper considers a networked control system, as depicted in Figure \ref{fig:LQG model}. Specifically, the system has an LQG plant, $k$ sensors/encoders, and one decoder. The $k$ encoders receive signals of the plant's output and transmit a message independently through a lossless channel using prefix codes. Based on the messages, the decoder outputs a control signal to the plant. Two critical metrics are rates and the control cost. Rates measure the expected average bit length of the encoders' codewords, and the control cost measures the performance. The details can be found in Section \ref{section:problem-formulation}.

The rate-cost tradeoff for the single encoder and the single decoder/controller and the LQG/Gaussian-Markov source has been extensively studied \cite{GP1973,GP1974,lqgwithchannel,DO2012,SDP-GM,LQGSDP,lqgwithsensor,LQGsideinfo,lqgsideinfo2, CCKS2022, SS2024,causal-source-coding-scalar,1.254-gap,vector-achievable-scheme}. 
Gorbunov and Pinsker \cite{GP1973,GP1974} introduced the nonanticipatory epsilon entropy for the information-theoretic sequential coding problem and characterized it for the Gaussian-Markov source. Tatikonda et al. \cite{lqgwithchannel} studied the SRDF for communication-constrained LQG control. The operational achievable schemes with rates within 1.254 bits of the SRDF were studied by Derpich and {\O}stergaard\cite{DO2012} for the stationary Gaussian-Markov source and Silva et al. \cite{1.254-gap} for the LTI control system. Tanaka et al. \cite{vector-achievable-scheme} generalized the achievable scheme to the vector case. For computing the SRDF, Tanaka et al.\cite{SDP-GM} constructed the semidefinite programming form and established an equivalent single-letter form for the time-invariant system. This was later generalized to the LQG plant by Tanaka et al.\cite{LQGSDP}. In \cite{LQGSDP}, they also extended the idea to the partially observable plant. Charalambous et al. \cite{CCKS2022} further characterized the Gaussian test channels for the Gauss–Markov sources and established an equivalent state-dependent formulation. Stavrou and Skoglund\cite{SS2024} studied the SRDF for the partially observable vector Gaussian-Markov source. Kostina and Hassibi \cite{lqgwithsensor} applied the entropy power inequality to analyze the SRDF and proved an analytical lower bound for the problem.
For the decoder-only side information setting, Sabag et al.\cite{LQGsideinfo} established the optimality of linear policies and applied the SDP approach to compute the SRDF. Cuvelier and Tanaka \cite{lqgsideinfo2} considered the case where part of the plant state is the side information. Lev and Khina \cite{causal-source-coding-scalar} considered the scalar Gaussian-Markov source with side information. 

Few results are known for the setting with multiple sensors/encoders and one decoder/controller, as depicted in Figure \ref{fig:LQG model}. Stavrou et al. \cite{sum-rate-mimo} considered optimizing the sum rate with respect to the covariance matrix and mean-square error distortions. Their formulation of the sum rate optimization is closely related to the point-to-point case form. Jung et al. \cite{lqgceo} investigated the same scenario for a Gaussian-Markov source but with the feedback of the control signal sent from the decoder to the encoders. \cite{lqgceo} formulated the rate-distortion tradeoff as a non-convex optimization problem. In this work, we investigate the setting in Figure \ref{fig:LQG model} and establish new information-theoretic lower bounds. Our main contributions are summarized as follows:

\begin{itemize}
    \item We establish a new directed information lower bound for the LQG plant for the multiple sensors/encoders and one decoder setting in Theorem \ref{thm:dilb}.
    \item  We consider optimizing the new weighted sum directed information bound and show that it is sufficient to consider a class of linear, independent encoders and linear decoders under the full-observation condition for an LQG plant and quadratic cost in Theorem \ref{thm:gaussianmarkov}. 
    \item We apply our analysis to the point-to-point case with side information and recover the same structural result in \cite{LQGsideinfo} through an alternate proof in Theorem \ref{thm:SI-linear-policy}.
    \item We construct a new SDP form for computing the SRDF with side information. The proposed formulation accommodates a broader class of system parameters and recovers the existing SDP formulations in \cite{SDP-GM,LQGSDP,LQGsideinfo} in Theorem \ref{thm:sdp-tv}.
    \item We derive an equivalent single-letter SDP form for infinite-horizon time-invariant systems when all the modes of the state transition matrix on or inside the unit circle are controllable through the process noise in Theorem \ref{thm:sdp-ti}.
    \item We perform numerical experiments for the two-encoder setting with one full observation and one partial observation, and the case with SI as in Section \ref{sec:numerical}.     
\end{itemize}

\textbf{Notations.}
Lowercase letters are vectors and uppercase letters are matrices. Boldface lowercase letters denote random variables.
The subscripts represent time steps, and the superscripts represent the encoder index. For example, $x_{t}$ is a realization of a vector at time step $t$. $y^{i}_{t}$ is the signal received by Encoder $i$. 
We denote $x^{\T}$ as the transpose of $x$.
Let $(\rx,\ry)$ be two random variables on $(\mathcal{X}\times \mathcal{Y},\mathcal{B}_{\mathcal{X}}\otimes\mathcal{B}_{\mathcal{Y}})$ and let $\mathbb{P}_{\rx,\ry}$ and $\mathbb{G}_{\rx,\ry}$ be two joint measures. If $\mathbb{P}_{\rx|y}\ll\mathbb{G}_{\rx|y}$ holds $\mathbb{P}_{\ry}$ a.s., the conditional KL divergence is defined as 
\begin{align*}
    \KL_{\mathbb{P},\mathbb{G}} (\rx|\ry) =
    \int \log\left(\frac{\dP(x|y)}{\dG(x|y)}\right)\dP(x|y)\dP(y)
\end{align*}
and $\infty$ otherwise.
Consider three random variables $(\rx,\ry,\rz)$ with joint measure $\mathbb{P}$. If $\mathbb{P}_{\rx,\ry|z}\ll \mathbb{P}_{\rx|z}\otimes\mathbb{P}_{\ry|z}$ holds $\mathbb{P}_{\rz}$ a.s., the conditional mutual information is defined as
\begin{align*}
    I(\rx;\ry|\rz) = \int \log\left(\frac{\dP(x,y|z)}{\dP(x|z )\otimes \dP(y|z)}\right)\dP(x,y,z)
\end{align*}
and $\infty$ otherwise. We denote $[a,b] = \{a,...,b\}$ and $[t] = \{1,...,t\}$ where $a$, $b$ and $t$ are integers. Let $(\rx_{[T]},\ry_{[T]},\rz_{[T]})$ be a group of random variables. Define the directed information and Kramer's notation\cite{Kramer2003} for causally conditioned directed information as follows:
\begin{align*}
    I(\rx_{[T]}\rightarrow \ry_{[T]}) &=\sum_{t=1}^{T} I(\rx_{[t]};\ry_{t}|\ry_{[t-1]})\\
    I(\rx_{[T]}\rightarrow \ry_{[T]}\Vert \rz_{[T]})&=\sum_{t=1}^{T} I(\rx_{[t]};\ry_{t}|\ry_{[t-1]},\rz_{[t]})
\end{align*}
Let $S\subseteq \{1,...,k\}$ be a set of indices. We denote $x^{S}_{t} = [x^{i}_{t}:i\in S ]$. Furthermore, define the notation $(x^{S}_{[t]})^{\T}=[ (x^{i}_{[t]})^{\T}:i\in S ]$.

\textbf{Organization.}
The rest of the paper is organized as follows. In Section \ref{section:problem-formulation}, we formulate the problem setting of the networked control system. In Section \ref{section:DILB}, we present the new directed information lower bound for the multiple sensors/encoders and one decoder setting. In Section \ref{sec:linearpolicyoptimal}, we show that a class of linear policies suffices to optimize the weighted-sum directed information lower bound. In Section \ref{sec:SI-policy}, we apply the analysis to the single encoder and single decoder setting with side information (SI) and show the structural result for the SRDF.
In Section \ref{sec:SDP-TV}, we derive the SDP formulation for the time-varying system, and in Section \ref{sec:SDP-TI}, we derive the equivalent single-letter SDP form for the infinite-horizon time-invariant system. The numerical experiments are in Section \ref{sec:numerical}.

\section{Problem Formulation}\label{section:problem-formulation}

Consider the setting with $k$ sensors/encoders and one decoder as depicted in Figure \ref{fig:LQG model}. The plant follows the process 
\begin{equation}\label{eqn:control-equation}
    \rx_{t+1} = A_{t}\rx_{t}+B_{t}\ru_{t}+\rw_{t}
\end{equation}
 where $x_{t}\in\mathbb{R}^{n}$ are the plant states, $u_{t}\in \mathbb{R}^{m}$ are the control signals, $w_{t}\in\mathbb{R}^{n}$
 are independent Gaussian noises following the distribution $\mathcal{N}(0,W_{t})$, $A_{t}\in\mathbb{R}^{n\times n}$ and $B_{t}\in\mathbb{R}^{n\times m}$. We assume the initial state $\rx_1$ follows $\mathcal{N}(0,P_{1|0})$, where $P_{1|0}\succ 0$. 

Let $i\in[k]$. At each time $t$, the sensor $i$ receives a linear signal $\ry_t^i$:
 \begin{equation}\label{eqn:linear-sensor}
     \ry_t^i = C^{i}_{t}\rx_{t}+ \rv^{i}_{t}
 \end{equation}
 where $C_t^{i}\in\mathbb{R}^{l^{i}_t \times n}$, $\rv^{[k]}_t$ are Gaussian noises independent of $(\rx_{[t]},\ry_{[t-1]}^{[k]},\ru_{[t-1]},\rw_{[t
 ]})$. Each $\rv^{i}_t$ follows distribution $\mathcal{N}(0,V^{i}_{t})$, where $V_t^{i}\in\mathbb{R}^{l^{i}_t\times l^{i}_t}$. Note that the sensor noises $\rv^{i}_{t}$ may be correlated across sensors, and their covariance matrices may be singular. 
  
\begin{figure}[t]
\centering
\begin{tikzpicture}[
    font=\small,
    >={Stealth[length=2mm]},
    block/.style={draw, thick, align=center, inner sep=3pt},
    line/.style={->, thick},
    dline/.style={->, thick, dashed},
  ]
  \node[block, minimum height=10mm] (plant)
    {Plant:\\[1pt] $x_{t+1}=A_t x_t + B_t u_t + w_t$};
  \node[block, right=5mm of plant, inner sep=2pt] (delay) {delay};

  \node[block, below=11mm of plant.south east, anchor=north,
        xshift=-2mm] (enc1)
    {$\text{Enc}_1$:\\ $\mathbb{P}\big(a_t^1 \,|\, y_{[t]}^1, a_{[t-1]}^1\big)$};
  \node[block, below=15mm of enc1] (enck)
    {$\text{Enc}_k$:\\ $\mathbb{P}\big(a_t^k \,|\, y_{[t]}^k, a_{[t-1]}^k\big)$};

  \node[block, right=6mm of enc1, minimum height=9mm] (s1) {Sensor\\ $1$};
  \node[block, right=6mm of enck, minimum height=9mm] (sk) {Sensor\\ $k$};

  \path (enc1) -- (enck) coordinate[midway] (midE);
  \node[block, left=10mm of midE -| enc1.west, anchor=east] (dec)
    {Dec:\\ $\mathbb{P}\big(u_t \,|\, u_{[t-1]}, a_{[t]}^{[k]}\big)$};

  \coordinate (bus) at ($(s1.east)+(4mm,0)$);
  \draw[line] (plant.east) -- (delay.west);
  \draw[thick] (delay.east) -| (bus |- sk.east);
  \draw[line] (bus |- s1.east) -- (s1.east);
  \draw[line] (bus |- sk.east) -- (sk.east);
  \fill (bus |- s1.east) circle (1.2pt);

  \draw[line] (s1.west) -- (enc1.east)
    node[midway, above=5mm, xshift=4mm] {$y_t^1 = C_t^1 x_t + v_t^1$};
  \draw[line] (sk.west) -- (enck.east)
    node[midway, above=5mm, xshift=4mm] {$y_t^k = C_t^k x_t + v_t^k$};

  \coordinate (ch) at ($(enc1.west)!0.5!(dec.east |- enc1.west)$);
  \draw[thick, dashed] (enc1.west) -- (ch |- enc1.west)
    node[pos=0.5, above] {$a_t^1$};
  \draw[thick, dashed] (enck.west) -- (ch |- enck.west)
    node[pos=0.5, above] {$a_t^k$};
  \draw[thick, dashed] (ch |- enc1.west) -- (ch |- enck.west);
  \draw[dline] (ch |- dec.east) -- (dec.east);

  \draw[line] ($(dec.north)+(-6mm,0)$) |- (plant.west)
    node[pos=0.25, left] {$u_t$};

  \node at ($(enc1)!0.42!(enck)$) {$\vdots$};
  \node at ($(s1)!0.42!(sk)$) {$\vdots$};
\end{tikzpicture}
\caption{Linear quadratic Gaussian model with k linear sensors. The dashed line implies a limited-rate lossless channel.}
    \label{fig:LQG model}
\end{figure}
 
 Let $T$ be the total number of steps. The encoders and the decoder operate as follows:
 \begin{itemize}
     \item \textbf{Encoder} $i$: at each time step, encoder $i$ observes $y_t^{i}$, and then generates a codeword $a^{i}_t$ with the local history information $(y_{[t]}^{i},a^{i}_{[t-1]})$ only. We assume the use of a prefix code\cite{cover2012elements}. Denote $\mathcal{A}^{i}_t \subset \{0,1,00,01,\dots\}$ as the alphabet of the prefix code for codeword $a^{i}_t$. The encoder $i$ is characterized by a stochastic kernel:
     \begin{align}\label{eqn:encoder-i}
         \mathbb{P}(a^{i}_{[T]}\Vert y^{i}_{[T]}) \triangleq \prod_{t=1}^{T} \mathbb{P}(a^{i}_t | y^{i}_{[t]},a^{i}_{[t-1]})
     \end{align}
The codeword $a_t^{i}$ is then transmitted through a lossless channel to the decoder;
\item \textbf{Decoder}: at each time step, the decoder/controller takes the $a_t^{[k]}$ and then outputs the control signal $u_{t}$ with the stochastic kernel:
\begin{align}\label{eqn:controller}
    \mathbb{P}(u_{[T]}\Vert a^{[k]}_{[T]}) \triangleq \prod_{t=1}^{T} \mathbb{P}(u_t|u_{[t-1]},a^{[k]}_{[t]})
\end{align}
 \end{itemize}
 

Denote the feasible policy set $\Gamma'$ as the tuple $\left\{( \mathbb{P}(a^{i}_{[T]}\Vert y^{i}_{[T]}),\mathcal{A}^{i}_{[T]})_{i=1}^{k},\mathbb{P}(u_{[T]}\Vert a^{[k]}_{[T]})\right\}$.
Let $l(*)$ denote the codeword length in bits. The expected average codeword length and the control cost are defined as follows:
\begin{defn}[Rates]
\begin{align*}
    R_{i} &= \frac{1}{T}\sum_{t=1}^{T} \mathbb{E}[l(\ra^{i}_{t})], \quad i \in [k]
\end{align*}
\end{defn}

\begin{defn}[Quadratic Control Cost]
    \begin{align*}
    J(\rx_{[T+1]},\ru_{[T]}) = \frac{1}{T} \sum_{t=1}^{T}\mathbb{E}[\Vert \rx_{t+1}\Vert_{Q_{t}}^{2} + \Vert \ru_{t}\Vert_{R_{t}}^{2}],
\end{align*} 
for given $Q_{t}\succ 0$ and $R_{t}\succ 0$, where $\Vert x\Vert_{Q}^{2} = x^{T}Qx$.
\end{defn}

Let $d^{*}$ be the minimal control cost such that the rates are finite. For the control cost requirement $d> d^{*}$, we are interested in the following weighted sum rate optimization problem:
\begin{align}\label{eqn:rate-distortion}
    \inf_{\gamma' \in \Gamma', J\leq d}   \sum_{i=1}^{k}\lambda_{i}R_{i}
\end{align}
where we assume $\lambda_{i+1}\geq \lambda_{i}\geq 0$ for $i\in [k-1]$ without loss of generality, as we can relabel the encoders. 




\section{Lower Bound And The Sufficient Linear Policy}
In this section, we derive a lower bound on $\sum_{i\in S}R_{i}$ for any $S\subseteq [k]$. Then we consider the lower bound of \eqref{eqn:rate-distortion} and show that the linear policy set $\Gamma_1$ is sufficient to optimize it.


\subsection{Lower Bound}\label{section:DILB}
For the multi-encoder and single-decoder problem setting as in Section \ref{section:problem-formulation}, we establish the following lower bound for any $\sum_{i\in S}R_{i}$.

\begin{thm}[Lower Bound]\label{thm:dilb}
    For any policy $\gamma' \in \Gamma'$, and any $S\subseteq [k]$, the rate term $\sum_{i\in S}R_{i}$ must satisfy the following inequality:
    \begin{align*}
        \sum_{i\in S}R_{i}\geq \frac{1}{T} I(\ry^{S}_{[T]} \rightarrow \ru_{[T]}\Vert \ry^{\bar{S}}_{[T]})
     \end{align*}
     where $\bar{S} = [k]\backslash S$.
\end{thm}
\begin{proof}
    For any $S\subseteq [k]$, we have
    \begin{align*}
       \sum_{i\in S} \mathbb{E}[l(\ra^{i}_{t})] 
        &\geq \sum_{i\in S}H(\ra^{i}_{t})\\
        &\geq H(\ra^{S}_{t}) \\
            &\geq H(\ra^{S}_{t}|\ra^{S}_{[t-1]},\ru_{[t-1]},\ry^{\bar{S}}_{[t]}) \\
            &\geq I(\ry^{S}_{[t]};\ra^{S}_{t}|\ra^{S}_{[t-1]},\ru_{[t-1]},\ry^{\bar{S}}_{[t]})
    \end{align*}
where the first inequality follows from the use of a prefix code; the second and the third inequalities follow from the fact that conditioning reduces entropy.
    Next, we aim to prove
    \begin{align}\label{align:lowerboundinequality}
        \sum_{t=1}^{T}I(\ry^{S}_{[t]};\ra^{S}_{t}|\ra^{S}_{[t-1]},\ru_{[t-1]},\ry^{\bar{S}}_{[t]})\geq \sum_{t=1}^{T}I(\ry^{S}_{[t]};\ru_{t}|\ru_{[t-1]},\ry^{\bar{S}}_{[t]})
    \end{align}

    Define 
    \begin{align*}
        \phi_{t} = I(\ry^{S}_{[t]};\ra^{S}_{t}|\ra^{S}_{[t-1]},\ru_{[t-1]},\ry^{\bar{S}}_{[t]}) - I(\ry^{S}_{[t]};\ru_{t}|\ru_{[t-1]},\ry^{\bar{S}}_{[t]}).
    \end{align*}
    Our goal is to show that $\sum_{t=1}^{T}\phi_{t} \geq 0$. Applying the chain rule twice to $I(\ry^{S}_{[t]};\ra^{S}_{[t]},\ru_{t}|\ru_{[t-1]},\ry^{\bar{S}}_{[t]})$ gives:

    \begin{align}
        &I(\ry^{S}_{[t]};\ra^{S}_{[t-1]}|\ru_{[t-1]},\ry^{\bar{S}}_{[t]}) + I(\ry^{S}_{[t]};\ra^{S}_{t},\ru_{t}|\ru_{[t-1]},\ry^{\bar{S}}_{[t]},\ra^{S}_{[t-1]})\nonumber\\
        =&I(\ry^{S}_{[t]};\ru_{t}|\ru_{[t-1]},\ry^{\bar{S}}_{[t]}) + I(\ry^{S}_{[t]};\ra^{S}_{[t]}|\ru_{[t]},\ry^{\bar{S}}_{[t]})\label{eqn:lb-eq2}
    \end{align}
        By the encoder $\eqref{eqn:encoder-i}$ and the controller \eqref{eqn:controller}, we have that $\ry^{S}_{[t]}-(\ra^{S}_{[t]},\ru_{[t-1]},\ry^{\bar{S}}_{[t]})-\ru_{t}$ is Markov and hence $I(\ry^{S}_{[t]};\ru_{t}|\ra^{S}_{[t]},\ru_{[t-1]},\ry^{\bar{S}}_{[t]}) = 0$.
    Therefore,
    \begin{align}
        &I(\ry^{S}_{[t]};\ra^{S}_{t}|\ru_{[t-1]},\ry^{\bar{S}}_{[t]},\ra^{S}_{[t-1]})\nonumber\\
        =&I(\ry^{S}_{[t]};\ra^{S}_{t},\ru_{t}|\ru_{[t-1]},\ry^{\bar{S}}_{[t]},\ra^{S}_{[t-1]})\label{eqn:lb-eq1}
        \end{align}
    Through manipulation, one can show the following:
    \begin{align*}
        \phi_{t} = & I(\ry^{S}_{[t]};\ra^{S}_{t}|\ru_{[t-1]},\ry^{\bar{S}}_{[t]},\ra^{S}_{[t-1]}) - I(\ry^{S}_{[t]};\ru_{t}|\ru_{[t-1]},\ry^{\bar{S}}_{[t]})\\
        \stackrel{(a)}{=} &I(\ry^{S}_{[t]};\ra^{S}_{[t]}|\ru_{[t]},\ry^{\bar{S}}_{[t]}) - I(\ry^{S}_{[t]};\ra^{S}_{[t-1]}|\ru_{[t-1]},\ry^{\bar{S}}_{[t]})\\
        \stackrel{(b)}{=} &I(\ry^{S}_{[t]};\ra^{S}_{[t]}|\ru_{[t]},\ry^{\bar{S}}_{[t]}) - I(\ry^{S}_{[t-1]};\ra^{S}_{[t-1]}|\ru_{[t-1]},\ry^{\bar{S}}_{[t]}) \\
        &- I(\ry^{S}_{t};\ra^{S}_{[t-1]}|\ru_{[t-1]},\ry^{\bar{S}}_{[t]},\ry^{S}_{[t-1]})\\
        \stackrel{(c)}{=} &I(\ry^{S}_{[t]};\ra^{S}_{[t]}|\ru_{[t]},\ry^{\bar{S}}_{[t]}) - I(\ry^{S}_{[t-1]};\ra^{S}_{[t-1]}|\ru_{[t-1]},\ry^{\bar{S}}_{[t]})
    \end{align*}
    where $(a)$ is due to \eqref{eqn:lb-eq2} and \eqref{eqn:lb-eq1}; 
    $(b)$ is by applying the chain rule to $I(\ry^{S}_{[t]};\ra^{S}_{[t-1]}|\ru_{[t-1]},\ry^{\bar{S}}_{[t]})$; $(c)$ holds since 
    $\ry^{S}_{t}-(\ru_{[t-1]},\ry^{\bar{S}}_{[t]},\ry^{S}_{[t-1]})-\ra^{S}_{[t-1]}$ is Markov, which
    is due to \eqref{eqn:control-equation}, \eqref{eqn:linear-sensor}, \eqref{eqn:encoder-i} and \eqref{eqn:controller}.

    To establish $\sum_{t=1}^{T}\phi_{t} \geq 0$, it remains to show the following inequality and apply a telescoping sum argument. 
    \begin{align*}
        I(\ry^{S}_{[t-1]};\ra^{S}_{[t-1]}|\ru_{[t-1]},\ry^{\bar{S}}_{[t]}) \leq I(\ry^{S}_{[t-1]};\ra^{S}_{[t-1]}|\ru_{[t-1]},\ry^{\bar{S}}_{[t-1]})
    \end{align*}
    To establish the above inequality, we apply the chain rule twice to $I(\ry^{S}_{[t-1]},\ry^{\bar{S}}_{t};\ra^{S}_{[t-1]}|\ru_{[t-1]},\ry^{\bar{S}}_{[t-1]})$ first:
    \begin{align*}
        &I(\ry^{S}_{[t-1]},\ry^{\bar{S}}_{t};\ra^{S}_{[t-1]}|\ru_{[t-1]},\ry^{\bar{S}}_{[t-1]})\\
        \stackrel{(i)}{=} &I(\ry^{S}_{[t-1]};\ra^{S}_{[t-1]}|\ru_{[t-1]},\ry^{\bar{S}}_{[t-1]}) \\
        +&I(\ry^{\bar{S}}_{t};\ra^{S}_{[t-1]}|\ru_{[t-1]},\ry^{\bar{S}}_{[t-1]},\ry^{S}_{[t-1]})\\
        \stackrel{(ii)}{=} &I(\ry^{\bar{S}}_{t};\ra^{S}_{[t-1]}|\ru_{[t-1]},\ry^{\bar{S}}_{[t-1]}) + I(\ry^{S}_{[t-1]};\ra^{S}_{[t-1]}|\ru_{[t-1]},\ry^{\bar{S}}_{[t]})
    \end{align*}
    where $(i)$ is by applying the chain rule to $\ry^{S}_{[t-1]}$ first and $(ii)$ is by applying the chain rule to $\ry^{\bar{S}}_{t}$ first.
    
    The above manipulations of the chain rule give us the following:
    \begin{align*}
        &I(\ry^{S}_{[t-1]};\ra^{S}_{[t-1]}|\ru_{[t-1]},\ry^{\bar{S}}_{[t]})  - I(\ry^{S}_{[t-1]};\ra^{S}_{[t-1]}|\ru_{[t-1]},\ry^{\bar{S}}_{[t-1]})\\
        = &I(\ry^{\bar{S}}_{t};\ra^{S}_{[t-1]}|\ru_{[t-1]},\ry^{\bar{S}}_{[t-1]},\ry^{S}_{[t-1]})
        \\
        &- I(\ry^{\bar{S}}_{t};\ra^{S}_{[t-1]}|\ru_{[t-1]},\ry^{\bar{S}}_{[t-1]})
        \leq 0 
    \end{align*}
    where the inequality holds since $\ra^{S}_{[t-1]} - (\ru_{[t-1]},\ry^{\bar{S}}_{[t-1]},\ry^{S}_{[t-1]})-\ry^{\bar{S}}_{t}$ is Markov, which is due to \eqref{eqn:control-equation}, \eqref{eqn:linear-sensor}, \eqref{eqn:encoder-i} and \eqref{eqn:controller}.
    The telescoping part is as follows:
    \begin{align*}
        &\sum_{t=1}^{T}\phi_{t}\\ =&\sum_{t=1}^{T}I(\ry^{S}_{[t]};\ra^{S}_{[t]}|\ru_{[t]},\ry^{\bar{S}}_{[t]}) - I(\ry^{S}_{[t-1]};\ra^{S}_{[t-1]}|\ru_{[t-1]},\ry^{\bar{S}}_{[t]})\\
        \geq &\sum_{t=1}^{T}I(\ry^{S}_{[t]};\ra^{S}_{[t]}|\ru_{[t]},\ry^{\bar{S}}_{[t]}) - I(\ry^{S}_{[t-1]};\ra^{S}_{[t-1]}|\ru_{[t-1]},\ry^{\bar{S}}_{[t-1]})\\
        =&I(\ry^{S}_{[T]};\ra^{S}_{[T]}|\ru_{[T]},\ry^{\bar{S}}_{[T]})\\
        \geq & 0
    \end{align*}
\end{proof}
Different decompositions of the weighted sum rate can lead to different lower bounds. For example, applying Theorem \ref{thm:dilb} directly to each $R_{i}$ term leads to the following optimization problem:
\begin{align*}
    \inf_{\gamma \in \Gamma^{*}, J\leq d}  
     &\sum_{i=1}^{k}\lambda_{i} I(\ry^{\{i\}}_{[T]}\rightarrow \ru_{[T]} \Vert \ry^{[k]\backslash \{i\}}_{[T]})
\end{align*}
where $\Gamma^{*}\triangleq \{ \mathbb{P}(u_{t}|y^{[k]}_{[t]},u_{[t-1]}),\eqref{eqn:control-equation}, \eqref{eqn:linear-sensor}\}_{t\in[T]}$.
Among the possible decompositions, we decompose the weighted sum rate by successively grouping the largest possible subset of encoders with the largest admissible nonnegative coefficient for the corresponding sum-rate term. This yields $\sum_{i=1}^{k}\lambda_{i}R_{i} = \sum_{i=1}^{k}(\lambda_{i} - \lambda_{i-1}) (\sum_{j=i}^{k}R_{j})$ with $\lambda_{0} = 0$. By Theorem \ref{thm:dilb}, we obtain the following optimization problem:
\begin{align}\label{eqn:weighted sum lower bound}
    \inf_{\gamma \in \Gamma^{*}, J\leq d}  
     &\sum_{i=1}^{k}(\lambda_{i} - \lambda_{i-1}) I(\ry^{[i,k]}_{[T]}\rightarrow \ru_{[T]} \Vert \ry^{[i-1]}_{[T]})
\end{align}
We adopt this particular decomposition because, under the full-observation condition, we can show in Theorem \ref{thm:gaussianmarkov} that it is sufficient to consider a class of linear policies for \eqref{eqn:weighted sum lower bound} with independent encoders, and that the last encoder can be chosen to be memoryless.
\begin{remark}\label{remark:suboptimality}
    The lower bound in Theorem \ref{thm:dilb} may not be tight in general. For example, consider the two-encoder case. It follows from the proof of Theorem \ref{thm:dilb} that $\frac{1}{T}I(\ry^{1}_{[T]}\rightarrow \ru_{[T]} \Vert \ra^{2}_{[T]})$ can be a strictly tighter lower bound for $R_{1}$. Nevertheless, we are interested in the lower bound of Theorem \ref{thm:dilb} because it is numerically more tractable, as demonstrated in Theorem \ref{thm:gaussianmarkov}.
\end{remark}
\begin{remark}
    The lower bound in Theorem \ref{thm:dilb} is valid for the special case when only one encoder has limited rate, and the rest of the encoders have unlimited rate. In this case, the lower bound can be viewed as the setting with side information to the decoder only, as in \cite{LQGsideinfo,lqgsideinfo2}. In Section \ref{sec:SI-policy}, we give an alternate proof for the structural result of the lower bound.
\end{remark}

\subsection{Linear Policies Are Sufficient}\label{sec:linearpolicyoptimal}
In this subsection, we show that a class of linear policies is sufficient for \eqref{eqn:weighted sum lower bound} under the full-observation condition, thus reducing the original problem to a finite-dimensional optimization problem. Denote the collection of all the sensor signals $y'_{t} \triangleq \begin{bmatrix}
    y^{1}_{t}\\
    \vdots\\
    y^{k}_{t}\end{bmatrix}$.
The full-observation condition is:
\begin{condition}[Full Observation]\label{con:fullobservation}
   There exists a matrix $C'_{t}$ such that $C'_{t}y'_{t} =x_{t}$.
\end{condition}
The full-observation condition implies that one can recover $x_{t}$ from all the linear sensors. This condition is needed when proving Lemma \ref{lemma:P-equal-G} for the following Theorem:

\begin{thm}[Sufficient Linear Policy]\label{thm:gaussianmarkov}
If the full-observation condition \ref{con:fullobservation} is satisfied, then it is sufficient to consider the following policy set $\Gamma_{1}$ for \eqref{eqn:weighted sum lower bound}:
   \begin{itemize}
       \item $\rc^{j}_{t} = F^{j}_{t}\ry^{j}_{[t]}$ for $j \in [k-1]$;
       \item $\rc^{k}_{t} = F^{k}_{t}\ry^{k}_{t} + \rh_{t}$ where $\rh_{t}\sim \mathcal{N}(0,H_{t})$ are mutually independent and $H_t\succeq 0$; 
   \item $\ru_{t} =  \sum_{i=1}^{k}\rc^{i}_{t}+N_{t}\ru_{[t-1]}$.
   \end{itemize}
  The dimensions of $(\rc^{i}_{t}:i\in[k])$ are the same as the dimension of $\ru_{t}$.
   The optimization for \eqref{eqn:weighted sum lower bound} is reduced to finding the finite-dimensional matrices $\{H_{t},N_{t},F^{i}_{t},i\in[k]\}_{t=1}^{T}$.
\end{thm}

\begin{proof}
    See Section \ref{proof:gaussianmarkov}.
\end{proof}
\begin{remark}
   Theorem \ref{thm:gaussianmarkov} reduces \eqref{eqn:weighted sum lower bound} to a finite-dimensional optimization problem, enabling numerical evaluation of the proposed lower bound. In Section \ref{sec:numerical}, we use this formulation to investigate the rate-control cost tradeoff and compare our bound with existing bounds.
\end{remark}

\subsection{Application to Side Information Case}\label{sec:SI-policy}

In this section, we consider the single encoder and the single decoder setting with side information (SI) at the decoder only. This setting can be viewed as a special case of the two-encoder setting, with $\ry_t^{1} = \rx_t$, $\ry_t^{2}\triangleq \ry_t = C_t \rx_t +\rv_t$, in which encoder 2 has unlimited rate, but encoder 1 has limited rate, as in Figure \ref{fig:SI model}.
In this case, we are interested in optimizing the $TR_1$ lower bound:
 \begin{equation}\label{eqn:SI-opt}
    \inf_{\gamma\in\Gamma^{*}, J\leq d}  I(\rx_{[T]}\rightarrow \ru_{[T]}\Vert \ry_{[T]})
\end{equation}
\begin{figure}[H]
    \centering
    \includegraphics[scale = 0.8]{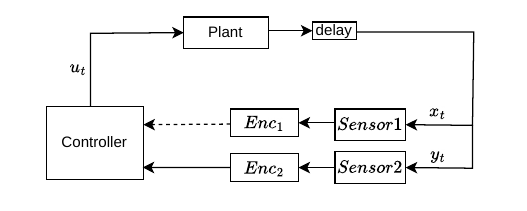}
    \caption{Two-encoder case with one full observation and one partial observation. The dashed line implies a limited-rate channel. The encoder 2 transmits through an unlimited-rate channel.}
    \label{fig:SI model}
\end{figure} 

Following the proof of Theorem \ref{thm:gaussianmarkov}, we show that it is sufficient to consider linear encoders and the certainty equivalence controller for minimizing the conditional directed information lower bound.
The optimal control gain $K_t$ for the certainty equivalence controller can be computed via the following backward Riccati recursions.
\begin{align}
    S_{t} &= \begin{cases} 
      Q_t & \text{if } t=T \\
      Q_t + \Phi_{t+1}  & \text{if } t<T
   \end{cases}\\
   \Phi_t &= A_t^{\T}(S_t - S_tB_t(B_t^{\T}S_tB_t+R_t)^{-1}B_t^{\T}S_t)A_t\\
   K_t &= -(B_t^{\T}S_t B_t +R_t)^{-1}B_t^{\T}S_t A_t\\
   \Theta_t &= K_t^{\T}(B_t^{\T}S_t B_t+R_t)K_t
\end{align}
\begin{thm}\label{thm:SI-linear-policy}
Consider the point-to-point setting with the encoder observing $\rx_t$ and side information $\ry_t = C_t \rx_t + \rv_t$, where $\rv_t \sim \mathcal{N}(0,V_t)$ with $V_t \succeq 0$, at the controller at each time step. The optimization problem \eqref{eqn:SI-opt} is equivalent to
\begin{align*}
    \inf_{\gamma\in\Gamma, J\leq d}\sum_{t=1}^{T}I(\rx_t;\rc_t|\rc_{[t-1]},\ry_{[t]})
\end{align*}
where the policy set $\Gamma$ consists of:
\begin{itemize}
       \item $\rc_{t} = F_{t}\rx_{t} + \rh_{t}$ where $\rh_{t}$ is independent of $(\rx_{[t]},\ry_{[t]},\rh_{[t-1]})$ and follows $\mathcal{N}(0,H_{t})$ with $H_t\succeq 0$;
   \item $\ru_{t} =  K_t \mathbb{E}[\rx_t|\ry_{[t]} , \rc_{[t]}, \ru_{[t-1]}]$.
   \end{itemize}
\end{thm}
\begin{proof}
    The proof is the same as that for Theorem \ref{thm:gaussianmarkov} till the end of Section \ref{subsec:linear-policyproof-node1}. At that step, we are considering the linear encoder in $\Gamma$ and the linear controller. Notice that the objective is independent of the choice of controller. Therefore, we choose the optimal controller, the certainty equivalence controller. By the choice of encoder in the policy set $\Gamma$, $\rc_t - \rx_t -(\rx_{[t-1]},\rc_{[t-1]},\ru_{[t-1]},\ry_{[t]})$ is Markov. Moreover, the certainty equivalence controller is deterministic. Therefore, we have
    \begin{align*}
        I(\rx_t;\rc_t|\rc_{[t-1]},\ry_{[t]})
        =I(\rx_{[t]};\rc_t|\ru_{[t-1]},\rc_{[t-1]},\ry_{[t]}).
    \end{align*}
    To conclude the sufficiency, the last step follows the proof of Theorem \ref{thm:dilb}, which shows $I(\rx_{[T]}\rightarrow \rc_{[T]}\Vert \ru_{[T-1]},\ry_{[T]}) \geq I(\rx_{[T]}\rightarrow \ru_{[T]}\Vert \ry_{[T]})$.
\end{proof}
\begin{remark}
In \cite{LQGsideinfo}, Sabag et al. optimize $I(\rx_{[T]}\rightarrow \rc_{[T]}\Vert \ry_{[T]})$ with respect to the control cost constraint and the stochastic kernels $\mathbb{P}(c_{[T]}\Vert x_{[T]})\triangleq \prod_{t=1}^{T}\mathbb{P}(c_t |x_{[t]},c_{[t-1]})$ and $\mathbb{P}(u_{[T]}\Vert c_{[T]},y_{[T]})\triangleq \prod_{t=1}^{T}\mathbb{P}(u_t|u_{[t-1]},c_{[t]},y_{[t]})$. They show that the optimal policy set is $\Gamma$ as in Theorem \ref{thm:SI-linear-policy}. Here, we show that optimizing a seemingly different objective \eqref{eqn:SI-opt} results in the same structural result. 
\end{remark}

\section{Semidefinite Programming for the SI Case}\label{section:LQG-SI}

In this section, we present a semidefinite programming (SDP) approach to compute the lower bound \eqref{eqn:weighted sum lower bound} for the decoder-only side information case. Unfortunately, we are currently unaware of how it can be extended to compute the lower bound \eqref{eqn:weighted sum lower bound} for multi-sensor cases.
Given the structural result in Theorem \ref{thm:SI-linear-policy}, the objective of \eqref{eqn:SI-opt} can be written as a linear combination of log determinants. Previous works \cite{LQGSDP,SDP-GM,LQGsideinfo,SS2024} show that the resulting optimization problems admit SDP forms, which enable efficient computation
for vector systems. Moreover, the optimal variable can be used to construct Gaussian test channels.

We extend the existing results in two aspects. First, we relax the conditions under which an SDP representation of the lower bound is available. Specifically, we show that such an SDP formulation exists when $A_t A_t^{\T} + W_t \succ 0$ and $V_t \succ 0$. The construction of the SDP form makes use of the null space and the Schur complement. 
Second, we derive a single-letter representation of the lower bound for infinite-horizon, time-invariant systems. Let $A_t = A$, $W_t = W$ for all $t$ and decompose $W= LL^{\T}$. We show that if $(A,L)$ is controllable for the modes on or inside the unit circle, the infinite-horizon optimization problem admits an equivalent single-letter form. We note that \cite[Theorem 4]{SS2024} gives a similar single-letter characterization. However, under the conditions listed in \cite[Theorem 4]{SS2024}, one can construct an example such that the log-determinant optimization form in \cite[Theorem 4]{SS2024} is infeasible.

\subsection{Time-Varying System}\label{sec:SDP-TV}
In this subsection, we consider a finite-horizon time-varying system and show that an SDP form exists for the point-to-point case with side information $\ry_t = C_t \rx_t + \rv_t$ and $V_t\succ 0$. We assume $Q_t\succ 0$ and $R_t\succ 0$, and $A_t A_t^{\T} + W_t\succ0$. The conditions $A_t A_t^{\T} + W_t\succ0$ and $V_t\succ 0$ are to guarantee that both the unconditional and posterior supports of $\rx_t$ are $\mathbb{R}^{n}$ for all $t$, assuming $P_{1|0}\succ 0$. Our SDP form generalizes existing forms due to its weaker constraint on $A_t$ and $W_t$. 

Denote $\rx_{t|t} \triangleq \mathbb{E}[\rx_{t}|\rc_{[t]},\ry_{[t]}]$, $\rx_{t+1|t} \triangleq \mathbb{E}[\rx_{t+1}|\rc_{[t]},\ry_{[t]}]$, and $\rx_{t+1|t}^{\SI}\triangleq \mathbb{E}[\rx_{t+1}|\rc_{[t]},\ry_{[t+1]}]$. The error covariance matrices are defined as:
\begin{align*}
    P_{t|t} &= \mathbb{E}[ (\rx_t-\rx_{t|t})(\rx_t-\rx_{t|t})^{\T} ]\\
    P_{t+1|t} &= \mathbb{E}[ (\rx_{t+1}-\rx_{t+1|t})(\rx_{t+1}-\rx_{t+1|t})^{\T} ]\\
    P_{t+1|t}^{\SI} &= \mathbb{E}[(\rx_{t+1}-\rx_{t+1|t}^{\SI})(\rx_{t+1}-\rx_{t+1|t}^{\SI})^{\T} ]
\end{align*}
The following equalities follow from Kalman filtering theory:
\begin{align}
    P_{t+1|t} &= A_{t}P_{t|t}A_{t}^{\T} + W_{t} \label{eqn:kalman-update}\\
    (P_{t+1|t}^{\SI})^{-1} &= P_{t+1|t}^{-1} + C_{t+1}^{\T}V_{t+1}^{-1}C_{t+1} \label{eqn:sideinfo-equation}
\end{align}

\begin{thm}[SDP Time-Varying]\label{thm:sdp-tv}Let $W_t = L_t L_t^{\T}$ with the reduced-rank factorization, $r_t = \rank L_t$, $N_t^{\T}  \triangleq [E_t, D_t]$ whose rows form an orthonormal basis of $\ker [A_t,L_t]$ with $E_t\in\mathbb{R}^{r_t\times n}, D_t\in\mathbb{R}^{r_t\times r_t}$, $\Delta_t = E_t P_{t|t}E_t^{\T} + D_t D_t^{\T}$ and $\Xi_t = E_t P_{t|t}A_t^{\T}+D_tL_t^{\T}$. Then \eqref{eqn:SI-opt} equals $\max\{0,\eqref{eqn:sdp-tv}\}$.

\begin{subequations}\label{eqn:sdp-tv}
\begin{align}
    \inf_{\{\Pi_{t},P_{t|t}\}_{t\in[T]}} &c_{0}-\frac{1}{2}\log|P_{T|T}|- \sum_{t=1}^{T-1}\left( \frac{1}{2}\log|\Pi_{t}|\right.\nonumber\\
    &\left.+\frac{1}{2}\log|V_{t+1}+C_{t+1}P_{t+1|t}C_{t+1}^{\T}|\right)\\
 \text{s.t.}&\qquad \Pi_{t}\succ 0 \label{sdp-tv-c1}\\
&\qquad P_{t+1|t} = A_{t}P_{t|t} A_{t}^{\T} + L_{t}L_{t}^{\T} \label{sdp-tv-c2}  \\
&\qquad  P_{T|T} = \Pi_{T}, \quad P_{1|1} \preceq P_{1|0}^{\SI} \label{sdp-tv-c3}\\
&\qquad \begin{bmatrix}
        P_{t|t-1}-P_{t|t} & P_{t|t-1}C_{t}^{\T}\\
        C_{t}P_{t|t-1} & C_{t}P_{t|t-1}C_{t}^{\T} + V_{t}
    \end{bmatrix}\succeq 0
    \label{sdp-tv-c4}\\
&\qquad \begin{bmatrix}
        \Delta_t - \Pi_t & \Xi_t\\
       \Xi_t^{\T} & P_{t+1|t}
    \end{bmatrix}\succeq 0
    \label{sdp-tv-c5}\\
    &\qquad \sum_{t=1}^{T}\Tr(\Theta_t P_{t|t}) + \eta_1 \leq Td \label{sdp-tv-c6}
\end{align}
\end{subequations}
with
\begin{align*}
    c_{0} &=  \frac{1}{2}\log|P_{1|0}^{\SI}|+\sum_{t=1}^{T-1}\frac{1}{2}(\log|V_{t+1}| + \log|A_{t}A_t^{\T}+W_t|)\\
    \eta_1 &=\sum_{t=1}^{T}\Tr(S_t W_t)  + \Tr(\Phi_1 P_{1|0})
\end{align*}
and \eqref{sdp-tv-c2}\eqref{sdp-tv-c5} hold for $t\in[T-1]$ and \eqref{sdp-tv-c1}\eqref{sdp-tv-c4} hold for $t\in[T]$.
Moreover, the Gaussian test channel $c_{t} = F_{t} x_{t} + h_{t}$ with $h_{t}\sim \mathcal{N}(0,H_{t})$ can be found through singular value decomposition of $ P_{t|t}^{-1} - (P_{t|t-1}^{\SI})^{-1}$:
\begin{align*}
    P_{t|t}^{-1} - (P_{t|t-1}^{\SI})^{-1} = F_{t}^{\T}H_{t}^{-1}F_{t}
\end{align*}
where rank($H_{t}$) = rank($P_{t|t}^{-1} - (P_{t|t-1}^{\SI})^{-1}$).
\end{thm}
\begin{proof}
    From the proof of Theorem \ref{thm:SI-linear-policy}, it is sufficient to consider $\inf_{\gamma\in\Gamma, J\leq d}\sum_{t=1}^{T}I(\rx_t;\rc_t|\rc_{[t-1]},\ry_{[t]})$. Moreover, since $V_t\succ 0$, following the same steps as \cite{LQGSDP}, we restrict $H_t \succ 0$. By \eqref{eqn:kalman-update}, $P_{t+1|t}\succ 0 $ if $P_{t|t}\succ 0$. Given $P_{1|0}\succ 0$ and $A_t A_t^{\T} + W_t\succ 0$ and combining \eqref{eqn:sideinfo-equation} and 
\begin{align*}
  P_{t+1|t+1}^{-1}  =   (P_{t+1|t}^{\SI})^{-1} + F_{t+1}^{\T} H_{t+1}^{-1}F_{t+1},
\end{align*}
we conclude that $P_{t|t}\succ 0$ for all $t\in [T]$. To show necessity, suppose $A_t A_t^{\T} + W_t$ is not full rank. Then there exists a non-zero $z\in\mathbb{R}^{n}$ such that $A_t^{\T}z = W_t z = 0$. Consequently, $A_t P_{t|t} A_{t}^{\T} + W_t$ is singular even if $P_{t|t}$ is full rank.

Then we show how to get the SDP form. Note that the objective can be written as
\begin{align*}
    &\sum_{t=1}^{T}I(\rx_t;\rc_t|\rc_{[t-1]},\ry_{[t]})\\
    =&\sum_{t=1}^{T}\frac{1}{2}\left(\log|P_{t|t-1}^{\SI}| - \log|P_{t|t}|\right)\\
    =&\sum_{t=1}^{T-1}\frac{1}{2}\left(\log|P_{t+1|t}^{\SI}| - \log|P_{t|t}|\right)
    +\frac{1}{2}\log|P_{1|0}^{\SI}| - \frac{1}{2}\log|P_{T|T}|
\end{align*}
By \cite[Theorem 18.1.1]{matrix-determinant}, $\log|P_{t+1|t}^{\SI}|$ equals
\begin{align*}
    \log|P_{t+1|t}| + \log|V_{t+1}| -\log|V_{t+1}+C_{t+1}P_{t+1|t}C_{t+1}^{\T}|
\end{align*}
Since $-\log|V_{t+1}+C_{t+1}P_{t+1|t}C_{t+1}^{\T}|$ is already a convex function in $P_{t|t}$, we consider $\log|P_{t+1|t}| - \log|P_{t|t}|$ next. Denote $r_t = \rank L_t$. Let $M_{t} = [A_t, L_t]\in\mathbb{R}^{n\times(n+r_t)}$ and the rows of $N_t^{\T} = [E_t, D_t]\in \mathbb{R}^{r_t\times(n+r_t)}$ be $\ker M_t$'s orthonormal basis. Then we have
\begin{align*}
    M_t N_t =0 \qquad N_t^{\T}N_t = I_{r_t}.
\end{align*}
Denote $Z_t = \begin{bmatrix}
    N_t^{\T}\\
    M_t
\end{bmatrix}$. Then $Z_t \begin{bmatrix}
    P_{t|t} &0\\
    0& I_{r_t}
\end{bmatrix}Z_t^{\T}$ equals
\begin{align*}
    \begin{bmatrix}
    E_t P_{t|t}E_t^{\T} +D_t D_t^{\T} & E_tP_{t|t}A_t^{\T} + D_tL_t^{\T}\\
    A_t P_{t|t} E_{t}^{\T}+L_t D_t^{\T} & A_t P_{t|t}A_{t}^{\T} + L_t L_t^{\T}
\end{bmatrix}
\triangleq &\begin{bmatrix}
    \Delta_t & \Xi_t \\
    \Xi^{\T}_t & P_{t+1|t}
\end{bmatrix}
\end{align*}
Since $Z_t Z_t^{\T} = \begin{bmatrix}
    I_{r_t} & 0 \\
    0 & A_{t}A_{t}^{\T} +L_{t} L_{t}^{\T}
\end{bmatrix} \succ 0$, by the Schur complement, $|P_{t+1|t}|$ equals
\begin{align*}
    \frac{\left|Z_t \begin{bmatrix}
    P_{t|t} &0\\
    0& I_{r_t}
\end{bmatrix}Z_t^{\T} \right|}{|\Delta_t - \Xi_t P^{-1}_{t+1|t}\Xi_t^{\T}|}
\end{align*}
Note that $Z_t$ is full rank; we have
\begin{align*}
    \left|Z_t \begin{bmatrix}
    P_{t|t} &0\\
    0& I_{r_t}
\end{bmatrix}Z_t^{\T} \right| &=|A_tA_t^{\T} +L_t L_t^{\T}| |P_{t|t}|.
\end{align*}
Hence, $\log|P_{t+1|t}| - \log|P_{t|t}|$ equals
\begin{align*}
& - \log|\Delta_t - \Xi_t P^{-1}_{t+1|t}\Xi_t^{\T}| +\log|A_t A_t^{\T} +W_t|\\
   =&  \inf_{\Pi_t}- \log|\Pi_t|+\log|A_t A_t^{\T} +W_t|\\
   &\text{s.t.    } \Pi_t\succ 0 \qquad \begin{bmatrix}
       \Delta_t - \Pi_t & \Xi_t \\
       \Xi_t^{\T} & P_{t+1|t}
    \end{bmatrix}\succeq 0
\end{align*}
where the equality is due to Schur's complement.

The semidefinite constraint $P_{t+1|t+1}\preceq P_{t+1|t}^{\SI}$ gives the constraint:
\begin{align*}
    \begin{bmatrix}
        P_{t+1|t} - P_{t+1|t+1} & P_{t+1|t}C_{t+1}^{\T}\\
        C_{t+1}P_{t+1|t} & V_{t+1} + C_{t+1}P_{t+1|t} C_{t+1}^{\T}
    \end{bmatrix}\succeq 0.
\end{align*}
The control cost constraint is $\sum_{t=1}^{T}\Tr(\Theta_t P_{t|t}) + \eta_1 \leq Td$ where $\eta_1 =\sum_{t=1}^{T}\Tr(S_t W_t)  + \Tr(\Phi_1 P_{1|0})$.
\end{proof}

\subsection{Time-Invariant System}\label{sec:SDP-TI}
In this subsection, we consider infinite-horizon time-invariant systems, i.e., the system parameters, like $A_{t} = A$, $W_{t} = W$, $C_{t} = C$, and $V_{t} = V$ are fixed for all $t$. We assume $V\succ 0$, $(A,B)$ is stabilizable, and that $(A,W)$ is controllable for the modes on or inside the unit circle. We consider the following optimization problem:
\begin{align}\label{eqn:opt-TI}
\inf_{
      \gamma \in \Gamma^{*},
      J\leq d
 }
     & \limsup_{T\rightarrow \infty} \frac{1}{T}I(\rx_{[T]}\rightarrow \ru_{[T]}\Vert \ry_{[T]})
\end{align}
In particular, the control cost constraint $J$ here is $\limsup_{T\rightarrow \infty}\frac{1}{T}\sum_{t=1}^{T}\mathbb{E}[\Vert \rx_{t+1}\Vert_{Q}^{2} + \Vert \ru_{t}\Vert_{R}^{2}]$, and we assume $Q\succ 0$, $R\succ 0$ and $d>d^{*}$. Since $(A,B)$ is stabilizable, the optimal control gain matrix $K_t$ converges and the limit can be computed by first finding the unique stabilizing solution to the following Riccati equation\cite[Theorem 14.5.3]{linearestimation}:
\begin{align}
    A^{\T}SA-S-A^{\T}SB(B^{\T}SB+R)^{-1}B^{\T}SA+Q = 0
\end{align}
Then the control gain matrix $K = -(B^{\T}SB+R)^{-1}B^{\T}SA$. Set $\Theta = K^{\T}(B^{\T}SB+R)K$.
We show that the above optimization problem admits an equivalent single-letter SDP formulation under the stated conditions.
\begin{thm}[SDP Time-Invariant]\label{thm:sdp-ti}
Let $W= LL^{\T}$ be a reduced-rank factorization and assume that $(A,L)$ is controllable for the modes on or inside the unit circle. Let $r = \rank L$, $N^{\T}  = [E, D]$ whose rows form an orthonormal basis of $\ker [A,L]$ with $E\in\mathbb{R}^{r\times n}, D\in\mathbb{R}^{r\times r}$, $\Delta = E PE^{\T} + D D^{\T}$, and $\Xi = E P A^{\T}+DL^{\T}$. Then \eqref{eqn:opt-TI} equals $\max\{0,\eqref{eqn:sdp-ti}\}$.
\begin{subequations}\label{eqn:sdp-ti}
\begin{align}
    \inf_{\{\Pi,P\}} &- \frac{1}{2}\log|V+C\Tilde{P}C^{\T}|- \frac{1}{2}\log|\Pi|\nonumber\\
    &+\frac{1}{2}\log|V| + \frac{1}{2}\log|AA^{\T}+W|\\
    \text{s.t.}&\qquad \Pi\succ 0, P\succ 0 \\
    &\qquad \begin{bmatrix}
        \Delta - \Pi & \Xi \\
        \Xi^{\T} & \Tilde{P}
    \end{bmatrix} \succeq 0 \\
    &\qquad \begin{bmatrix}
        \Tilde{P}-P & \Tilde{P}C^{\T}\\
        C\Tilde{P} & C\Tilde{P}C^{\T} + V
    \end{bmatrix}\succeq 0\\
    &\qquad \Tilde{P} = APA^{\T} + LL^{\T} \\
    &\qquad \Tr(P\Theta)+\Tr(SW)\leq d
\end{align}
\end{subequations}
    Moreover, for any feasible $P$, the stationary Gaussian test channel $c_{t} = F x_{t} + h$ with $h\sim \mathcal{N}(0,H)$ can be found through singular value decomposition of $P^{-1} - \Tilde{P}^{-1}-C^{\T}V^{-1}C$:
\begin{align*}
    P^{-1} - \Tilde{P}^{-1}-C^{\T}V^{-1}C = F^{\T}H^{-1}F
\end{align*}
where $\rank H  = \rank (P^{-1} - \Tilde{P}^{-1}-C^{\T}V^{-1}C)$.
\end{thm}
\begin{proof}
    See Section \ref{section:proof-sdp-ti}
\end{proof}

\section{Numerical Simulation}\label{sec:numerical}
\noindent\textit{1) Two-Encoder Lower Bound:}
We first consider the two-encoder setting. In particular, we consider the scalar process with $A = B = Q=R =1$, $T=2$ and $y^{1}_t = x_t + v_t$, $y^{2}_t = x_t$. We set $W_t = V_t = 1$ and $P_{1|0} = 1$. The optimization problem then reduces to
\begin{align*}
    \inf_{\gamma \in \Gamma_{1}, J\leq d}  
     &I(\rx_{[2]},\ry^{1}_{[2]}\rightarrow \ru_{[2]} ) + (\lambda-1)I(\rx_{[2]}\rightarrow \ru_{[2]} \Vert \ry^{1}_{[2]})
\end{align*}
The weighted sum in the other direction degenerates to the sum rate bound since the full observation encoder can simulate the policy of the partial observation encoder.

For this numerical setting, the minimum total control cost with the partial observation only is 3.85, so we consider  $Td=3.84$ with $T = 2$ and $d = 1.92$. We perform the numerical optimization with the fmincon solver with both interior-point and SQP in MATLAB\cite{MATLAB}. For each value of $\lambda$, we run both algorithms with the same 100 randomly generated initial points and report the best objective obtained. Let $R_{full}$ and $R_{partial}$ denote the total communication rates (rather than average) of the encoders with full and partial state observations, respectively. We compare the weighted sum lower bound and the cut-set bound for $R_{full}$, $R_{partial}$, and $R_{full}+R_{partial}$. The cut-set bound is the SRDF for the single encoder case.
The numerical simulation shows our bound is notably tighter than the simple cut-set bound, as in Figure \ref{fig:two_encoder_simulation}.
\begin{figure}[!htb]
    \centering
    \includegraphics[scale = 0.55]{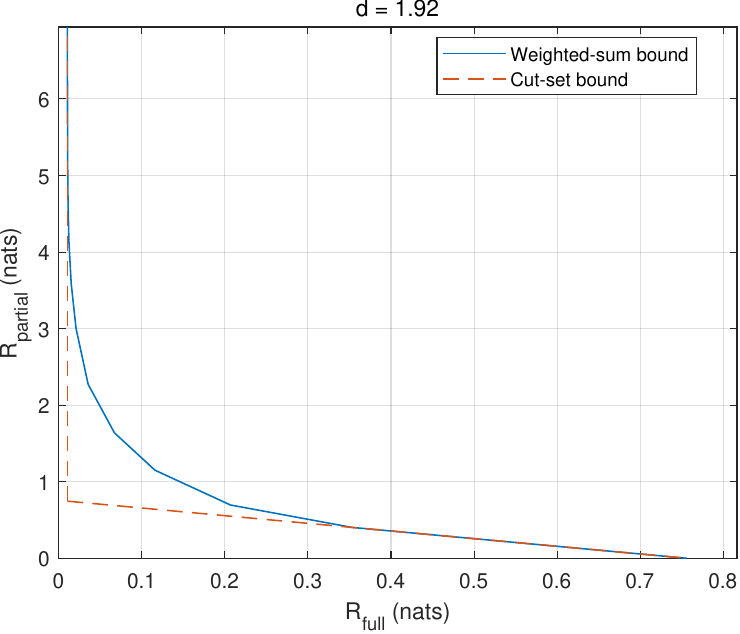}
    \caption{Numerical simulations of the two-encoder lower bound. Our lower bound is notably better than the simple cut-set bound.}
    \label{fig:two_encoder_simulation}
\end{figure}

\par\noindent\textit{2) SDP Form with Side Information:}
We then consider the SDP form in Theorem \ref{thm:sdp-ti} for the SI case. In particular, we consider the  following $A$ and $C$:
\begin{align*}
    A = \begin{bmatrix}
        2 & 0 & 0\\
        0 & 6 & 0\\
        0 & 0 & 0
    \end{bmatrix} \quad C = \begin{bmatrix}
        1 & 0 & 0
    \end{bmatrix}.
\end{align*}
We set $V = 1$, $B = Q = R  = I$ and consider two different process noise covariance matrices $W^{1}$ and $W^{2}$:
\begin{align*}
    W^{1} = \begin{bmatrix}
        0 & 0 & 0\\
        0 & 1 & 0\\
        0 & 0 & 1
    \end{bmatrix}\quad W^{2}=\begin{bmatrix}
        0 & 0 & 0\\
        0 & 0 & 0\\
        0 & 0 & 1
    \end{bmatrix}.
\end{align*}
In both cases, the lower bound approaches $\log 6$ as the control cost tends to infinity, consistent with the data-rate theorem \cite{NE2004}. Moreover, for $W^{1}$, the lower bound goes to infinity as the control cost approaches the theoretical minimum, whereas for $W^{2}$, the lower bound approaches a constant since neither unstable mode is driven by the process noise. 
Our numerical result captures this finite-limit behavior, which is not covered by existing results \cite{LQGsideinfo}. Although the limit is finite, it is not attained at the minimum control cost by any positive definite $P$ in Theorem \ref{thm:sdp-ti}. The simulations are performed with SDPT3\cite{SDPT3} with the YALMIP\cite{YALMIP} interface, as in Figure \ref{fig:numerical_SI}.

\begin{figure}[!htb]
    \centering
    \includegraphics[scale = 0.6]{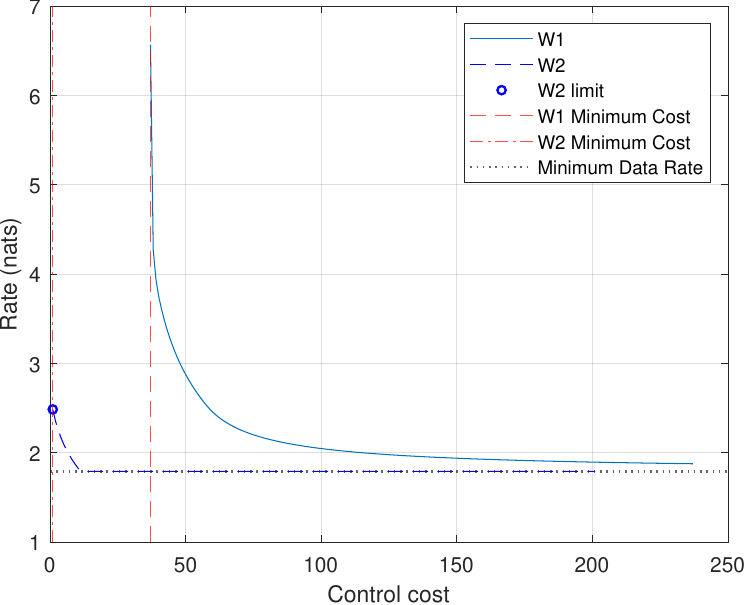}
    \caption{Numerical simulations of the SRDF for the SI case.}
    \label{fig:numerical_SI}
\end{figure}
\section{Proofs}

\subsection{Proof of Theorem \ref{thm:gaussianmarkov}}
\label{proof:gaussianmarkov}
Recall the policy set $\Gamma_{1}$:
    \begin{itemize}
       \item $\rc^{j}_{t} = F^{j}_{t}\ry^{j}_{[t]}$ for $j \in [k-1]$;
       \item $\rc^{k}_{t} = F^{k}_{t}\ry^{k}_{t} + \rh_{t}$ where $\rh_{t}\sim \mathcal{N}(0,H_{t})$;
   \item $\ru_{t} = \sum_{j=1}^{k}\rc^{j}_{t}+N_{t}\ru_{[t-1]}$.
   \end{itemize}
Consider the policy set $\Gamma_{2}$:
    \begin{equation}
    \Gamma_{2}: \ru_{t} = \sum_{j=1}^{k-1}M_{t}^{j}\ry^{j}_{[t]} + M_{t}^{k} \ry^{k}_{t} + N_{t} \ru_{[t-1]} + \rg_{t} \text{,   } \rg_{t}\sim \mathcal{N}(0,G_{t})
    \end{equation}
To prove Theorem \ref{thm:gaussianmarkov}, we aim to show the following inequalities:
\begin{align}
    &\inf_{\gamma\in\Gamma^{*}:J_{\gamma}\leq d}\sum_{i=1}^{k}(\lambda_{i} - \lambda_{i-1}) I(\ry^{[i,k]}_{[T]}\rightarrow \ru_{[T]} \Vert \ry^{[i-1]}_{[T]}) \nonumber \\
    \geq & \inf_{\gamma\in\Gamma^{*}:J_{\gamma}\leq d}\sum_{i=1}^{k}(\lambda_{i} - \lambda_{i-1})\nonumber \\&\sum_{t=1}^{T}I(\ry^{[i,k-1]}_{[t]},\ry^{k}_{t};\ru_{t}|\ru_{[t-1]},\ry^{[i-1]}_{[t]})\label{ineq:1lower}\\
    \geq& \inf_{\gamma\in \Gamma_{2}:J_{\gamma}\leq d}\sum_{i=1}^{k}(\lambda_{i} - \lambda_{i-1})\nonumber\\
    &\sum_{t=1}^{T}I(\ry^{[i,k-1]}_{[t]},\ry^{k}_{t};\ru_{t}|\ru_{[t-1]},\ry^{[i-1]}_{[t]}) \label{ineq:2lower}\\
    \geq& \inf_{\gamma\in \Gamma_{1}:J_{\gamma}\leq d}\sum_{i=1}^{k}(\lambda_{i} - \lambda_{i-1})\nonumber \\ 
    &\sum_{t=1}^{T}I(\ry^{[i,k-1]}_{[t]},\ry^{k}_{t};\sum_{j\in[i,k]}\rc^{j}_{t}|\ry^{[i-1]}_{[t]},(\sum_{j\in[i,k]}\rc^{j})_{[t-1]})\label{ineq:3lower}\\
    \geq& \inf_{\gamma\in\Gamma^{*}:J_{\gamma}\leq d}\sum_{i=1}^{k}(\lambda_{i} - \lambda_{i-1}) I(\ry^{[i,k]}_{[T]}\rightarrow \ru_{[T]} \Vert \ry^{[i-1]}_{[T]})\label{ineq:4lower}
\end{align}

Our goal is to first show that the policies in $\Gamma_{2}$ give a lower bound for the optimization \eqref{eqn:weighted sum lower bound}. This shows that the optimal $\ru_{t}$ is a linear function of $(\ry^{[k-1]}_{[t]},\ry^{k}_{t},\ru_{[t-1]})$. Then, by separating the optimal $\ru_{t}$ into independent parts, corresponding to each encoder, one can show that linear and independent encoders in $\Gamma_{1}$ give a lower bound to policies in $\Gamma_{2}$. The last inequality is straightforward by considering the policy set $\Gamma_{1}$. This completes our proof.

The first inequality \eqref{ineq:1lower} is because mutual information is non-negative. The rest of the proofs are here.

\subsubsection[Proof of inequality 31]{Proof of \eqref{ineq:2lower}}
Let $\gamma\in \Gamma^{*}$ be a feasible policy in $\Gamma^{*}$ and $\gamma_{\mathbb{Q}}\in \Gamma_{2}$ be a feasible policy in $\Gamma_{2}$. Let $\mathbb{P}(y^{[k]}_{[T+1]},u_{[T]})$ be the measure generated by $\gamma$, \eqref{eqn:control-equation} and \eqref{eqn:linear-sensor}. Following the arguments in \cite{LQGSDP}, we assume that $\mathbb{P}$ is zero-mean. 
Let $\mathbb{G}(y^{[k]}_{[T+1]},u_{[T]})$ be a Gaussian measure with the same covariance matrix as $\mathbb{P}(y^{[k]}_{[T+1]},u_{[T]})$. The following lemma holds.

\begin{lemma}\label{lemma:lowerboundGaussian1}
For $i\in [k]$, the following inequality holds
        \begin{align*}
     &\sum_{t=1}^{T}I_{\mathbb{P}}(\ry^{[i,k-1]}_{[t]},\ry^{k}_{t};\ru_{t}|\ru_{[t-1]},\ry^{[i-1]}_{[t]})\\
     \geq &\sum_{t=1}^{T}I_{\mathbb{G}}(\ry^{[i,k-1]}_{[t]},\ry^{k}_{t};\ru_{t}|\ru_{[t-1]},\ry^{[i-1]}_{[t]})
    \end{align*}
\end{lemma}

\begin{proof}
    We assume that $I_{\mathbb{P}}(\ry^{[i,k-1]}_{[t]},\ry^{k}_{t};\ru_{t}|\ru_{[t-1]},\ry^{[i-1]}_{[t]})$ is finite for $t\in[T]$. Denote $\KL_t\triangleq\KL_{\mathbb{P},\mathbb{G}}(\ry^{[i,k-1]}_{[t]},\ry^{k}_t|\ru_{[t-1]},\ry^{[i-1]}_{[t]}) $ and $\KL_t^{+}\triangleq\KL_{\mathbb{P},\mathbb{G}}(\ry^{[i,k-1]}_{[t]},\ry^{k}_t|\ru_{[t]},\ry^{[i-1]}_{[t]})$.
By Lemma \ref{lemma:IP-IG}(a), $\KL_t$ is finite for every $t\in[T]$.  Therefore, Lemma \ref{lemma:one-step} implies $I_{\mathbb{G}}(\ry^{[i,k-1]}_{[t]},\ry^{k}_{t};\ru_{t}|\ru_{[t-1]},\ry^{[i-1]}_{[t]})$ is finite and
\begin{align*}
    \KL_{t}^{+}
    =& I_{\mathbb{P}}(\ry^{[i,k-1]}_{[t]},\ry^{k}_{t};\ru_{t}|\ru_{[t-1]},\ry^{[i-1]}_{[t]})+\KL_t \\
    &- I_{\mathbb{G}}(\ry^{[i,k-1]}_{[t]},\ry^{k}_{t};\ru_{t}|\ru_{[t-1]},\ry^{[i-1]}_{[t]})
\end{align*}
for $t\in[T]$.
By Lemma \ref{lemma:IP-IG}(b) and the nonnegativity of KL divergence, we have $\KL_{t}^{+} \geq \KL_{t+1}$.
Summation over the difference of mutual information terms for all $t\in[T]$ implies that
\begin{align*}
    &\sum_{t=1}^{T}I_{\mathbb{P}}(\ry^{[i,k-1]}_{[t]},\ry^{k}_{t};\ru_{t}|\ru_{[t-1]},\ry^{[i-1]}_{[t]})\\
    &-\sum_{t=1}^{T}I_{\mathbb{G}}(\ry^{[i,k-1]}_{[t]},\ry^{k}_{t};\ru_{t}|\ru_{[t-1]},\ry^{[i-1]}_{[t]})\\
    \geq & \sum_{t=1}^{T}(\KL_{t+1} - \KL_{t})\\
     =&\KL_{T+1}-\KL_{1}\\
     \geq &0
\end{align*}
where the last inequality follows from $\mathbb{P}(y^{[k]}_1)= \mathbb{G} (y^{[k]}_1)$, $\KL_{\mathbb{P},\mathbb{G}}(\ry^{[i,k-1]}_{[1]},\ry^{k}_{1}|\ry^{[i-1]}_{[1]}) = 0$ and the nonnegativity of KL divergence.
\end{proof}
Let $\sum_{i=1}^{k-1}M_{t}^{i}\ry^{i}_{[t]} + M_{t}^{k} \ry^{k}_{t} + N_{t} \ru_{[t-1]}$ be the least-mean-square-error estimation(lmsee) of $\ru_{t}$ given $(\ry^{[k-1]}_{[t]},\ry^{k}_{t},\ru_{[t-1]})$ in $\mathbb{G}(y^{[k]}_{[T+1]},u_{[T]})$ and let $G_{t}$ be its error covariance matrix. Denote the conditional kernel $\mathbb{Q}(u_{t}|y^{[k-1]}_{[t]},y^{k}_{t},u_{[t-1]})$ as
\begin{align}
\mathcal{N}(\sum_{i=1}^{k-1}M_{t}^{i}y^{i}_{[t]} + M_{t}^{k} y^{k}_{t} + N_{t} u_{[t-1]}, G_{t})
\end{align}
By construction, we have the following property
\begin{align}
    \dG(y^{[k-1]}_{[t]},y^{k}_{t},u_{[t]}) =& \dQ(u_{t}|y^{[k-1]}_{[t]},y^{k}_{t},u_{[t-1]})\nonumber \\
    &\cdot\dG(y^{[k-1]}_{[t]},y^{k}_{t},u_{[t-1]}) \label{eqn:Qproperty}
\end{align}
Define $\mathbb{Q}(y^{[k]}_{[T+1]},u_{[T]})$ with the given start distribution $\mathbb{Q}(y^{[k]}_{1})$ recursively by
\begin{align}
    \dQ(y^{[k]}_{[t]},u_{[t]})&=\dQ(u_{t}|y^{[k-1]}_{[t]},y^{k}_{t},u_{[t-1]})\dQ(y^{[k]}_{[t]},u_{[t-1]})\label{eqn:Q}\\
    \dQ(y^{[k]}_{[t+1]},u_{[t]})&=\dP(y^{[k]}_{t+1}|y^{[k]}_{t},u_{t})\dQ(y^{[k]}_{[t]},u_{[t]})\label{eqn:Qplus1}
\end{align}
where $\dP(y^{[k]}_{t+1}|y^{[k]}_{t},u_{t})$ follows \eqref{eqn:control-equation} and \eqref{eqn:linear-sensor}. Then we have the following identity between $\mathbb{G}$ and $\mathbb{Q}$.
\begin{lemma}\label{lemma:distributionidentity}
For $t\in[T]$,
\begin{align*}
    \dG(y^{[k-1]}_{[t+1]},y^{k}_{t+1},u_{[t]}) = \dQ(y^{[k-1]}_{[t+1]},y^{k}_{t+1},u_{[t]}).
\end{align*}
\end{lemma}
\begin{proof}
    We prove by induction. For the base case, $\mathbb{G}
(y^{[k]}_{1}) = \mathbb{Q}
(y^{[k]}_{1})$ by the input Gaussian distribution. Suppose the identity holds at time step $t-1$. We show it also holds at time step $t$.
\begin{align*}
    &\dQ(y^{[k-1]}_{[t+1]},y^{k}_{t+1},u_{[t]}) \\
    =& \int_{y^{k}_{t}}\dQ(y^{[k-1]}_{[t+1]},y^{k}_{t+1},u_{[t]},y^{k}_{t})\\
    \stackrel{(a)}{=}& \int_{y^{k}_{t}}\dP(y^{[k]}_{t+1}|y^{[k]}_{t},u_{t})
    \dQ(y^{k}_{t},y^{[k-1]}_{[t]},u_{[t]})\\
    \stackrel{(b)}{=}& \int_{y^{k}_{t}}\dP(y^{[k]}_{t+1}|y^{[k]}_{t},u_{t})
    \dQ(u_{t}|y^{[k-1]}_{[t]},y^{k}_{t},u_{[t-1]})\nonumber\\
    &\dQ(y^{[k-1]}_{[t]},y^{k}_{t},u_{[t-1]})\\
    \stackrel{(c)}{=}&\int_{y^{k}_{t}}\dP(y^{[k]}_{t+1}|y^{[k]}_{t},u_{t})
    \dQ(u_{t}|y^{[k-1]}_{[t]},y^{k}_{t},u_{[t-1]})\nonumber\\
    &\dG(y^{[k-1]}_{[t]},y^{k}_{t},u_{[t-1]})\\
    \stackrel{(d)}{=}&\int_{y^{k}_{t}}\dP(y^{[k]}_{t+1}|y^{[k]}_{t},u_{t})
    \dG(y^{[k-1]}_{[t]},y^{k}_{t},u_{[t]})\\
    \stackrel{(e)}{=}&\int_{y^{k}_{t}}
    \dG(y^{[k-1]}_{[t+1]},y^{k}_{t+1},u_{[t]},y^{k}_{t})\\
    =&\dG(y^{[k-1]}_{[t+1]},y^{k}_{t+1},u_{[t]})
\end{align*}
where (a) and (b) follow from \eqref{eqn:Qplus1} and \eqref{eqn:Q} respectively; (c) follows from the induction hypothesis; (d) follows from \eqref{eqn:Qproperty} and (e) follows from Lemma \ref{lemma:P-equal-G}.
\end{proof}
Since $\mathbb{P}$, $\mathbb{G}$ and $\mathbb{Q}$ have matching second moments for $(\ry^{[k-1]}_{[t+1]},\ry^{k}_{t+1},\ru_{[t]})$ and Condition \ref{con:fullobservation}, we have
\begin{align}
    J_{\gamma_{\mathbb{P}}}= J_{\gamma_{\mathbb{G}}}=J_{\gamma_{\mathbb{Q}}}
\end{align}
This completes our proof for inequality \eqref{ineq:2lower}.

\subsubsection[Proof of inequality 32]{Proof of \eqref{ineq:3lower}} \label{subsec:linear-policyproof-node1}

Given a policy $\gamma_{2}\in \Gamma_{2}$, we construct a policy $\gamma_{1}\in \Gamma_{1}$ such that $J_{\gamma_{1}} = J_{\gamma_{2}}$ and the objective functions are equal. Consider the following $\ru_t$ from $\gamma_{2}\in\Gamma_{2}$ 
\begin{equation}
    \ru_{t} = \sum_{j=1}^{k-1}M_{t}^{j}\ry^{j}_{[t]} + M_{t}^{k} \ry^{k}_{t} + N_{t} \ru_{[t-1]} + \rg_{t}
\end{equation}
Let $\rc^{j}_{t} = M_{t}^{j}\ry^{j}_{[t]}$ for $j\in[k-1]$ and $\rc^{k}_{t} =  M_{t}^{k} \ry^{k}_{t} + \rg_{t}$ where the covariance matrix for $\rg_t$ is $G_t$. We can write $\sum_{j=1}^{k}\rc^{j}_{t} = \ru_{t} - N_{t} \ru_{[t-1]}$, which implies
\begin{align*}
    &I(\ry^{[i,k-1]}_{[t]},\ry^{k}_{t};\ru_{t}|\ru_{[t-1]},\ry^{[i-1]}_{[t]})\\
    =& I(\ry^{[i,k-1]}_{[t]},\ry^{k}_{t};\sum_{j=1}^{k}\rc^{j}_{t}|(\sum_{j=1}^{k}\rc^{j})_{[t-1]},\ry^{[i-1]}_{[t]})\\
    =& I(\ry^{[i,k-1]}_{[t]},\ry^{k}_{t};\sum_{j=i}^{k}\rc^{j}_{t}|(\sum_{j=i}^{k}\rc^{j})_{[t-1]},\ry^{[i-1]}_{[t]})
\end{align*}
By setting $F_{t}^{j} = M_{t}^{j}$ for $j\in[k]$, $H_t = G_t$ and $\ru_t = \sum_{j=1}^{k}\rc^{j}_{t}+N_t\ru_{[t-1]}$, we construct policy $\gamma_1 \in \Gamma_1$. Since the objective and the control cost are the same for policies $\gamma_1$ and $\gamma_2$, we complete the proof for inequality \eqref{ineq:3lower}.

\subsubsection[Proof of inequality 33]{Proof of \eqref{ineq:4lower}}
Note the $u_{t}$ in policy set $\Gamma_{1}$ is $\ru_{t} = \sum_{j=1}^{k}\rc^{j}_{t}+N_{t}\ru_{[t-1]}$. 
Let $\gamma_{1}\in\Gamma_{1}$ and construct $\gamma\in\Gamma^{*}$ by borrowing the $\ru_{t}$ of $\gamma_{1}$ and $\rc_t^{j} = \rc^{j}_t$ for $j\in[k]$. Then the measure of $(\ry^{[k]}_{[T+1]},\ru_{[T]})$ is the same for the two policies. We have the same control cost. For the objective function, we have:
\begin{align*}
    &I(\ry^{[i,k-1]}_{[t]},\ry^{k}_{t};\sum_{j=i}^{k}\rc^{j}_{t}|(\sum_{j=i}^{k}\rc^{j})_{[t-1]},\ry^{[i-1]}_{[t]})\\
    = &I(\ry^{[i,k-1]}_{[t]},\ry^{k}_{t};\ru_{t}|\ru_{[t-1]},\ry^{[i-1]}_{[t]})\\
    =&I(\ry^{[i,k]}_{[t]};\ru_{t}|\ru_{[t-1]},\ry^{[i-1]}_{[t]})
\end{align*}
The second equality follows because $\ru_{t}-(\ru_{[t-1]},\ry^{[k-1]}_{[t]},\ry^{k}_{t})-\ry^{k}_{[t-1]}$ is Markov.
This completes our proof for inequality \eqref{ineq:4lower}.

\subsection{Proof of Theorem \ref{thm:sdp-ti}}\label{section:proof-sdp-ti}
Following a similar argument as in the proof of Theorem \ref{thm:SI-linear-policy}, the policy set $\Gamma$ is sufficient for \eqref{eqn:opt-TI}. Moreover, given $V\succ 0$, we can restrict to $H_t\succ 0$. Secondly, since $(A,B)$ is stabilizable, the certainty equivalence controller converges to $K$. This results in the following form of \eqref{eqn:opt-TI}.
\begin{subequations}\label{eqn:TI-CS-limit}
\begin{align}
    \inf_{\{P_{t|t}\}} &\limsup_{T\rightarrow \infty}\frac{1}{T}\sum_{t=1}^{T}\frac{1}{2}\left(\log|AP_{t-1|t-1}A^{\T}+W|+\log|V|\right.\nonumber \\
    & \left.-\log|V+C(AP_{t-1|t-1}A^{\T}+W)C^{\T}|-\log|P_{t|t}|\right)\\
    \text{s.t.}&\qquad \limsup_{T\rightarrow \infty} \frac{1}{T}\sum_{t=1}^{T} \left[\Tr(P_{t|t}\Theta)+\Tr(S W)\right]\leq d \label{TI-CS-limit-C1}\\
    &\qquad  \begin{bmatrix}
        P_{t|t-1} -P_{t|t} & P_{t|t-1}C^{\T}\\
        CP_{t|t-1} & CP_{t|t-1}C^{\T} + V
    \end{bmatrix}\succeq 0 \label{TI-CS-limit-C2}\\
    &\qquad P_{t|t-1} = AP_{t-1|t-1}A^{\T} + W\\
    &\qquad P_{t|t}\succ 0
\end{align}
\end{subequations}

Define the following set $\mathcal{D}_{0}$:
\begin{align*}
    \mathcal{D}_{0} = \{&P\succeq 0|\Tr(P\Theta) + \Tr(SW)\leq d \\
   &\begin{bmatrix}
        APA^{\T} + W -P & (APA^{\T} + W )C^{\T}\\
        C(APA^{\T} + W ) & C(APA^{\T} + W )C^{\T} + V
    \end{bmatrix}\succeq 0 \}
\end{align*}
By \cite[Lemma 1]{limit-existence}, $\mathcal{D}_{0}$ is bounded. Let $\{P_{t|t}\}_{t\in\mathbb{N}}$ be a sequence of covariance matrices such that the objective in \eqref{eqn:TI-CS-limit} is bounded. Define $\bar{P}_{T} = \frac{1}{T}\sum_{t=1}^{T}P_{t|t}$. Following the arguments in \cite[Lemma 2]{limit-existence}, there exists a subsequence $\{T_{i}\}_{i\in\mathbb{N}}$ such that
\begin{align*}
    \lim_{i\rightarrow \infty}\bar{P}_{T_{i}} = \bar{P}_{\infty}\in\mathcal{D}_{0}.
\end{align*}
Denote $\mathcal{D}^{+}_{0} = \{P\succ 0 : P\in\mathcal{D}_{0}\}$ and consider sensor 
  $ \begin{bmatrix}
        C\\
        I
    \end{bmatrix}\rx_t + e_t$ where $e_t\sim \mathcal{N}(0,\text{diag}(V,I))$.
Since the $(A,\begin{bmatrix}
    C\\
    I
\end{bmatrix})$ is detectable and $(A,L)$ is controllable for all the modes on or in the unit circle, by \cite[Theorem 3.2(D), Theorem 4.1]{kalman-convergence-conditions}, there exists $P\succ 0$ that satisfies the semidefinite constraint in $\mathcal{D}_{0}$. We can scale such $P$ to show the existence of $P^{*}\in\mathcal{D}_{0}^{+}$.
For $0<\epsilon<1$ and each $i$, denote 
\begin{align*}
    P^{\epsilon}_{T_{i}} = (1-\epsilon)\bar{P}_{T_{i}} + \epsilon P^{*}.
\end{align*}

Define 
\begin{align}
    f_{r,T} \triangleq&  \frac{1}{2}\log|V| + \frac{1}{T}\sum_{t=1}^{T} \frac{1}{2}\left(\log|AP_{t|t}A^{\T}+W|\right.\nonumber \\
    & \left.-\log|V+C(AP_{t|t}A^{\T}+W)C^{\T}|-\log|P_{t|t}|\right)\\
    f_{s}(P) \triangleq&\frac{1}{2}\log|V|+\frac{1}{2}\left(\log|APA^{\T}+W|\right.\nonumber \\
    & \left.-\log|V+C(APA^{\T}+W)C^{\T}|-\log|P|\right)
\end{align}
 Following the convex formulation in Theorem \ref{thm:sdp-tv} and the fact that partial minimization preserves convexity, $f_{s}(P)$ is a convex function. By Jensen's inequality, we have that 
\begin{align}
    f_{s}(\bar{P}_{T_{i}})&\leq f_{r,T_{i}}\label{eqn:fs-less-frT}\\
    f_{s}(P_{T_{i}}^{\epsilon}) &\leq (1-\epsilon)f_s(\bar{P}_{T_{i}}) + \epsilon f_s(P^{*})\label{eqn:fs-epsilon}
\end{align}
Let $f_{r}$ be the objective value of \eqref{eqn:TI-CS-limit}.
We have the following lemma.
\begin{lemma}\label{lemma:limsup}
$\limsup_{i\rightarrow \infty}f_{r,T_{i}} \leq f_{r}$
\end{lemma}
\begin{proof}
    Note that
    \begin{align*}
         &f_{r,T_{i}}\\
         =& \frac{1}{2}\log|V| + \frac{1}{T_{i}}\sum_{t=1}^{T_{i}} \frac{1}{2}\left(\log|AP_{t|t}A^{\T}+W|\right.\nonumber \\
    & \left.-\log|V+C(AP_{t|t}A^{\T}+W)C^{\T}|-\log|P_{t|t}|\right)\\
    =&\frac{1}{2}\log|V| + \frac{1}{T_{i}}\sum_{t=1}^{T_{i}} (\frac{1}{2}\log|AP_{t-1|t-1}A^{\T}+W|-\frac{1}{2}\log|P_{t|t}|\\
    &-\frac{1}{2}\log|V+C(AP_{t-1|t-1}A^{\T}+W)C^{\T}|)\\
    &+\frac{1}{2T_{i}}\log|AP_{T_{i}|T_{i}}A^{\T}+W|\\
    &-\frac{1}{2T_{i}}\log|V+C(AP_{T_{i}|T_{i}}A^{\T}+W)C^{\T}|\\
    &-\frac{1}{2T_{i}}\log|AP_{0|0}A^{\T}+W|\\
    &+\frac{1}{2T_{i}}\log|V+C(AP_{0|0}A^{\T}+W)C^{\T}|
    \end{align*}
    Since $-\log|V+C(AP_{T_{i}|T_{i}}A^{\T}+W)C^{\T}|\leq -\log|V|$, taking $\limsup_{i\rightarrow \infty}$ on $\frac{1}{2T_{i}}\log|AP_{T_{i}|T_{i}}A^{\T}+W|-\frac{1}{2T_{i}}\log|V+C(AP_{T_{i}|T_{i}}A^{\T}+W)C^{\T}|$ yields a non-positive value by \cite[Lemma 3]{limit-existence}. Hence, $\limsup_{i\rightarrow \infty} f_{r,T_{i}}\leq f_{r}$.

\end{proof}

Consider the following optimization problem
\begin{subequations}\label{eqn:TI-nolimit}
\begin{align}
   \inf_{P\succ 0} &\frac{1}{2}(\log|\Tilde{P}|+\log|V| -\log|P| \nonumber \\
    &\qquad -\log|V+C\Tilde{P}C^{\T}|)\\
    \text{s.t.}&\qquad  \Tr(P\Theta) + \Tr(SW)\leq d \\
    &\qquad  
       \begin{bmatrix}
        \Tilde{P} -P & \Tilde{P}C^{\T}\\
        C\Tilde{P} & C\Tilde{P}C^{\T} + V
    \end{bmatrix}\succeq 0 \\
    &\qquad \Tilde{P} = APA^{\T} + LL^{\T}
\end{align}
\end{subequations}
By \eqref{eqn:fs-less-frT}, \eqref{eqn:fs-epsilon} and Lemma \ref{lemma:limsup}, we have shown that for any $0<\epsilon<1$, there exists a sequence of $P_{T_{i}}^{\epsilon}\succ0$ such that
\begin{align*}
    \lim_{i\rightarrow\infty}f_{s}(P_{T_{i}}^{\epsilon})\leq (1-\epsilon) f_r +\epsilon f_s(P^{*})
\end{align*}
Letting $\epsilon$ tend to zero yields that \eqref{eqn:TI-nolimit} is a lower bound for \eqref{eqn:TI-CS-limit}.
Next we show that for any feasible $P$ for \eqref{eqn:TI-nolimit}, there exists $\{P_{t|t}\}_{t\in\mathbb{N}}$ such that $P_{t|t}$ converges to $P$ and for this set of $\{P_{t|t}\}_{t\in\mathbb{N}}$, \eqref{eqn:TI-CS-limit} equals \eqref{eqn:TI-nolimit}.

Let $P$ be a feasible choice of \eqref{eqn:TI-nolimit}.
Set $\Tilde{P} = APA^{\T}+LL^{\T}$ and $P^{-1}-\Tilde{P}^{-1} = C^{\T}V^{-1}C + F^{\T}H^{-1}F$. Denote $\Tilde{F}$ and $\Tilde{H}$ as:
\begin{align*}
    \Tilde{F} = \begin{bmatrix}
        C \\
        F
    \end{bmatrix}, \hspace{1 cm} \Tilde{H} = \begin{bmatrix}
        V & 0 \\
        0 & H
    \end{bmatrix}
\end{align*}

\begin{lemma} \label{lemma:Ricatti-Convergence}
If all modes on or inside the unit circle are controllable for $(A,L)$, then the constructed $(A,\Tilde{F})$ is detectable.
\end{lemma}

\begin{proof}
Set $\Tilde{K} = A\Tilde{P}\Tilde{F}^{\T}(\Tilde{H}+\Tilde{F}\Tilde{P}\Tilde{F}^{\T})^{-1}$. By the Woodbury matrix identity, $P^{-1}-\Tilde{P}^{-1} = \Tilde{F}^{\T}\Tilde{H}^{-1}\Tilde{F}$ implies
\begin{align*}
    (A-\Tilde{K}\Tilde{F})\Tilde{P}(A-\Tilde{K}\Tilde{F})^{\T}-\Tilde{P} = -LL^{\T}-\Tilde{K}\Tilde{H}\Tilde{K}^{\T}\preceq 0
\end{align*}
Next, we show $A-\Tilde{K}\Tilde{F}$ is Schur stable. Denote $x^{\mathrm{H}}$ as the conjugate transpose of $x$. Suppose $A-\Tilde{K}\Tilde{F}$ is not Schur stable. The above Lyapunov inequality implies that there exists an eigenvalue $|\lambda| = 1$ and a left eigenvector $x$ for $A-\Tilde{K}\Tilde{F}$ such that
\begin{align*}
    &x^{\mathrm{H}}(A-\Tilde{K}\Tilde{F})\Tilde{P}(A-\Tilde{K}\Tilde{F})^{\T}x-x^{\mathrm{H}}\Tilde{P}x\\
    =&(|\lambda|^{2}-1)x^{\mathrm{H}}\Tilde{P}x\\
    =&0\\
    =& - x^{\mathrm{H}}LL^{\T}x - x^{\mathrm{H}}\Tilde{K}\Tilde{H}\Tilde{K}^{\T}x
\end{align*}
Hence, $x^{\mathrm{H}}L = x^{\mathrm{H}}\Tilde{K} = 0$. This further implies
\begin{align*}
    x^{\mathrm{H}}(A-\Tilde{K}\Tilde{F}) = x^{\mathrm{H}}A = \lambda x^{\mathrm{H}}.
\end{align*}
So $x^{\mathrm{H}}$ is an eigenvector of $A$ with eigenvalue $|\lambda| = 1$ and $x^{\mathrm{H}}L = 0$. This contradicts the assumption that $(A,L)$ is unit-circle controllable. Hence, $A-\Tilde{K}\Tilde{F}$ is Schur stable. So $(A,\Tilde{F})$ is detectable.

\end{proof}
By \cite[Theorem 3.2(D), Theorem 4.1]{kalman-convergence-conditions}, there exists a sequence of $\{P_{t|t},P_{t|t-1}\}$ that satisfies the following recursion:
\begin{align*}
    P_{t|t-1} &= AP_{t-1|t-1}A^{\T} + W \\
    P_{t|t} &= P_{t|t-1}-P_{t|t-1}\Tilde{F}^{\T}(\Tilde{F}P_{t|t-1}\Tilde{F}^{\T}+\Tilde{H})^{-1}\Tilde{F}P_{t|t-1}
\end{align*}
with given initial $P_{1|0}\succ 0$ and converges to $(P,\Tilde{P})$. For the objective \eqref{eqn:TI-CS-limit}, we have 
\begin{align*}
     &\lim_{t\rightarrow \infty}I(\rx_{t};\rc_{t}|\rc_{[t-1]},\ry_{[t]})\\
    =&\lim_{t\rightarrow \infty}\frac{1}{2}\log|P_{t|t-1}^{\SI}| - \frac{1}{2}\log|P_{t|t}|\\
    =&\lim_{t\rightarrow \infty}  \frac{1}{2}(\log| P_{t|t-1}|+\log|V|\nonumber \\
    &\qquad -\log|V+C P_{t|t-1}C^{\T}|-\log|P_{t|t}|)\\
    =& \frac{1}{2}(\log|\Tilde{P}|+\log|V|-\log|V+C\Tilde{P}C^{\T}|-\log|P|)\nonumber \\
\end{align*}
And $\lim_{t\rightarrow \infty}\Tr(P_{t|t}\Theta) =\Tr(P\Theta)  \leq d -\Tr(SW)$. Therefore, by the Ces\`{a}ro mean, we have that 
\begin{align*}
    &\lim_{T\rightarrow \infty}\frac{1}{T}\sum_{t=1}^{T}I(\rx_{t};\rc_{t}|\rc_{[t-1]},\ry_{[t]}) = f_{s}(P)\\
    &\lim_{T\rightarrow \infty}\frac{1}{T}\sum_{t=1}^{T}\Tr(P_{t|t}\Theta) \leq d-\Tr(SW)
\end{align*}
This implies that \eqref{eqn:TI-CS-limit} is a lower bound for \eqref{eqn:TI-nolimit}. Following a similar argument in Theorem \ref{thm:sdp-tv}, we have the SDP form.

\section{Conclusion}

We studied the tradeoff between the communication rate and the control cost in the multiple sensors/encoders and one decoder setting. We established a new directed information lower bound and proved that, under the full-observation condition, a class of linear policies suffices to optimize the weighted sum rate lower bound. We applied the results to the decoder-only side information case and presented SDP formulations for a broader class of system parameters. We also derived an equivalent single-letter SDP formulation for the infinite-horizon time-invariant system. For future research directions, one direction is to try to show the sufficiency of the linear policies for the tighter directed information lower bound discussed in Remark \ref{remark:suboptimality}. 

\section*{Acknowledgement}
The authors acknowledge the use of OpenAI's GPT-6 for language editing. GPT-6 also suggested the key argument for the measure-theoretic Lemma \ref{lemma:one-step} and the null-space SDP construction in Theorem \ref{thm:sdp-tv}.
The authors carefully verified the mathematical correctness of the resulting proofs and take full responsibility for all statements and results in this paper.

\appendix

In this appendix, we present several useful lemmas.
\begin{lemma}\label{lemma:P-equal-G}
For any distribution $\mathbb{P}(y^{[k]}_{[T+1]},u_{[T]})$ generated by a policy in $\Gamma^{*}$ and its Gaussian version $\mathbb{G}(y^{[k]}_{[T+1]},u_{[T]})$, if the full-observation condition \ref{con:fullobservation} is satisfied, we have
        \begin{align*}
            \mathbb{G}(y^{[k]}_{t+1}|y^{[k-1]}_{[t]},y^{k}_{t},u_{[t]}) &= \mathbb{G}(y^{[k]}_{t+1}|y^{[k]}_{t},u_{t})\\
            &= \mathbb{P}(y^{[k]}_{t+1}|y^{[k]}_{t},u_{t})\\
            &= \mathbb{P}(y^{[k]}_{t+1}|y^{[k-1]}_{[t]},y^{k}_{t},u_{[t]})
        \end{align*}
\end{lemma}
\begin{proof}
    By the full-observation condition \ref{con:fullobservation}, there exists a matrix $C'_{t}$ such that $C’_{t}y'_{t} = x_{t}$. Since $(\rw_t,\rv^{[k]}_{t+1})$ are independent of $(\rx_{[t]},\ry^{[k]}_{[t]},\ru_{[t-1]})$, the Markov chain $\ry^{[k]}_{t+1} - (\ry^{[k]}_t,\ru_t) - (\ry^{[k-1]}_{[t-1]},\ru_{[t-1]})$ holds in $\mathbb{P}$. Moreover, by the full-observation condition \ref{con:fullobservation}, \eqref{eqn:control-equation} and \eqref{eqn:linear-sensor}, the sensor signals $y^{[k]}_{t+1}$ can be written as follows for a given realization of $(y^{[k]}_t,u_t)$:
\begin{align*}
\begin{bmatrix}
            C^{1}_{t+1}\\
            \vdots \\
            C^{k}_{t+1}
        \end{bmatrix} \left(A_{t}C'_{t}y^{[k]}_{t}+B_{t}u_{t}+\rw_{t}\right) +\begin{bmatrix}
            \rv^{1}_{t+1}\\
            \vdots \\
            \rv^{k}_{t+1}
        \end{bmatrix} 
\end{align*}
Therefore, by \cite[Lemma 5]{LQGSDP}, $\ry^{[k]}_{t+1} - (\ry^{[k]}_t,\ru_t) - (\ry^{[k-1]}_{[t-1]},\ru_{[t-1]})$ is also Markov in $\mathbb{G}$. Also, by \cite[Lemma 6]{LQGSDP}, the linear relationships \eqref{eqn:control-equation} and \eqref{eqn:linear-sensor} both hold in $\mathbb{G}$. These two points together establish the equalities we aim to prove.
\end{proof}

\begin{lemma}[One-Step Lemma]\label{lemma:one-step}
    Let $\mathbb{P}_{\rx,\ry,\rz}$ be a joint probability measure on $(\mathcal{X}\times \mathcal{Y}\times \mathcal{Z},\mathcal{B}_{\mathcal{X}}\otimes\mathcal{B}_{\mathcal{Y}}\otimes\mathcal{B}_{\mathcal{Z}} )$, where $\mathcal{X} = \mathbb{R}^{n}$, $\mathcal{Y} = \mathbb{R}^{m}$ and $\mathcal{Z} = \mathbb{R}^{l}$ and suppose $(\rx,\ry,\rz)$ have finite first and second moments under $\mathbb{P}$. Let $\mathbb{G}_{\rx,\ry,\rz}$ be the Gaussian measure with the same first and second moments as $\mathbb{P}_{\rx,\ry,\rz}$. If $I_{\mathbb{P}}(\rx;\ry|\rz)$ and $\KL_{\mathbb{P},\mathbb{G}}(\rx|\rz) $ are both finite, then $I_{\mathbb{G}}(\rx;\ry|\rz)$ and $\KL_{\mathbb{P},\mathbb{G}}(\rx|\ry,\rz)$ are both finite as well. Moreover,
    \begin{align*}
       I_{\mathbb{G}}(\rx;\ry|\rz)& =\int\log\left(\frac{\dG(x|y,z)}{\dG(x|z )}\right)\dG(x,y,z)\\
       \KL_{\mathbb{P},\mathbb{G}}(\rx|\ry,\rz) &= I_{\mathbb{P}}(\rx;\ry|\rz) + \KL_{{\mathbb{P},\mathbb{G}}}(\rx|\rz) -I_{\mathbb{G}}(\rx;\ry|\rz).
    \end{align*}
\end{lemma}
\begin{proof}
Since $I_{\mathbb{P}}(\rx;\ry|\rz)$ is finite, following the same step as \cite[Lemma 9(b)]{LQGSDP}, we have that
\begin{align*}
    &\mathbb{P}_{\rx|y,z}\ll\mathbb{P}_{\rx|z}, \mathbb{P}_{\ry,\rz}\text{ a.s., and}\\
    &I_{\mathbb{P}}(\rx;\ry|\rz) = \int\log\left(\frac{\dP(x|y,z)}{\dP(x|z)}\right)\dP(x,y,z)
\end{align*}
 Since $\KL_{\mathbb{P},\mathbb{G}}(\rx|\rz) $ is also finite, we have
 \begin{align*}
     \mathbb{P}_{\rx|z}\ll\mathbb{G}_{\rx|z}, \mathbb{P}_{\rz} \text{ a.s.}
 \end{align*}
 Together, we have
 \begin{align*}
 \mathbb{P}_{\rx|y,z}\ll\mathbb{P}_{\rx|z}\ll\mathbb{G}_{\rx|z}, \mathbb{P}_{\ry,\rz} \text{ a.s..}
 \end{align*}
 
 Next we show that $\mathbb{G}_{\rx|z} \sim \mathbb{G}_{\rx|y,z}$ holds $\mathbb{P}_{\ry,\rz}$ a.s. and $\mathbb{G}_{\ry,\rz}$ a.s. to have finite $I_{\mathbb{G}}(\rx;\ry|\rz)$ and $\KL_{\mathbb{P},\mathbb{G}}(\rx|\ry,\rz)$.
To show this, we first show that $\mathbb{G}_{\rx|z}$ and $\mathbb{G}_{\rx|y,z}$ have the same affine subspace support $\mathbb{P}_{\ry,\rz}$ a.s. and $\mathbb{G}_{\ry,\rz}$ a.s.. Let $\hat{x}_z$ and $\hat{x}_{yz}$ be the conditional means of $\rx$ given $\rz$ and $(\ry,\rz)$ in $\mathbb{G}$, respectively. The error covariance matrices are denoted as $\Sigma_z$ and $\Sigma_{yz}$, respectively. 

Firstly, we show that $\range\Sigma_z = \range \Sigma_{yz}$. By the conditioning, we have $0\preceq \Sigma_{yz}\preceq \Sigma_{z}$. Suppose that $\range\Sigma_{yz}\subsetneq \range\Sigma_z$, then there exists a vector $a\in \ker\Sigma_{yz}$ but not in $\ker \Sigma_{z}$. Therefore, $a^{\T}\Sigma_{yz}a = 0 < a^{\T}\Sigma_{z}a$. Due to the matching second moments between $\mathbb{P}$ and $\mathbb{G}$, we have
\begin{align*}
    &\mathbb{E}_{\mathbb{P}}[a^{\T}(\rx -\hat{x}_{\ry\rz})(\rx -\hat{x}_{\ry\rz})^{\T}a]\\
    =& \mathbb{E}_{\mathbb{G}}[a^{\T}(\rx -\hat{x}_{\ry\rz})(\rx -\hat{x}_{\ry\rz})^{\T}a]\\
    =&a^{\T}\Sigma_{yz}a = 0
\end{align*}
Denote the hyperplane $ H_{yz} = \{x: a^{\T}x= a^{\T}\hat{x}_{yz}\}$.
This implies that $\mathbb{P}(H_{yz}|y,z) = 1$ holds $\mathbb{P}_{\ry,\rz}$ a.s.. Note that $\mathbb{G}_{a^{\T}\rx|z} = \mathcal{N}(a^{\T}\hat{x}_{z},a^{\T}\Sigma_{z}a)$ is a nondegenerate scalar Gaussian distribution for any $z$, hence, $\mathbb{G}_{a^{\T}\rx|z}(a^{\T}\hat{x}_{yz}|z)=0 $ holds $\mathbb{P}_{\ry,\rz}$ a.s., which contradicts that $\mathbb{P}_{\rx|y,z}\ll \mathbb{G}_{\rx|z}$, $\mathbb{P}_{\ry,\rz}$ a.s.. So $\range\Sigma_{yz} =\range\Sigma_z$. 

Secondly, we show that $\hat{x}_{\ry\rz}-\hat{x}_{\rz} \in \range\Sigma_{z}$ both $\mathbb{G}_{\ry,\rz}$ a.s. and $\mathbb{P}_{\ry,\rz}$ a.s.. For any $a\in\ker\Sigma_{z}$, we have
\begin{align*}
    a^{\T}(x - \hat{x}_{yz}) = 0,  a^{\T}(x - \hat{x}_{z}) = 0 ,\mathbb{G}_{\rx,\ry,\rz}\text{ a.s. }
\end{align*}
Since this holds for any $a\in\ker\Sigma_{z}$ and $\ker\Sigma_{z}$ has finite dimension, then $\hat{x}_{yz}-\hat{x}_{z} \in \range\Sigma_{z}$, $\mathbb{G}_{\ry,\rz}$ a.s.. Due to the matching second moments, the same arguments imply $\hat{x}_{yz}-\hat{x}_{z} \in \range\Sigma_{z}$, $\mathbb{P}_{\ry,\rz} $-a.s..
This implies that $\mathbb{G}_{\rx|y,z}$ and $\mathbb{G}_{\rx|z}$ have the same affine subspace support both $\mathbb{G}_{\ry,\rz}$ a.s. and $\mathbb{P}_{\ry,\rz}$ a.s., hence we have
\begin{align*}
    \mathbb{G}_{\rx|z}\sim\mathbb{G}_{\rx|y,z}, \mathbb{P}_{\ry,\rz} \text{ a.s. and }\mathbb{G}_{\ry,\rz}\text{ a.s.}
\end{align*}
This also implies that $\mathbb{G}_{\rx,\ry|z}\ll\mathbb{G}_{\rx|z}\otimes \mathbb{G}_{\ry|z}$, $\mathbb{G}_{\rz}$ a.s., so
\begin{align*}
    I_{\mathbb{G}}(\rx;\ry|\rz)&= \int \log\left(\frac{\dG(x,y|z)}{\dG(x|z )\otimes \dG(y|z)}\right)\dG(x,y,z)\\
    &=\int\log\left(\frac{\dG(x|y,z)}{\dG(x|z )}\right)\dG(x,y,z)
\end{align*}
By \cite[Lemma 5.2.4]{gray2023entropy}, one can replace $\frac{\dG(x|y,z)}{\dG(x|z )}$ with the ratio of two Gaussian density functions with respect to the induced Lebesgue measure on the common affine support. The resulting log density ratio is a quadratic function and $\mathbb{G}$ has finite second moments. Therefore, $I_{\mathbb{G}}(\rx;\ry|\rz)$ is finite, and by the matching moments the following also holds:
\begin{align*}
    &\int\log\left(\frac{\dG(x|y,z)}{\dG(x|z )}\right)\dG(x,y,z)\\
    =&\int\log\left(\frac{\dG(x|y,z)}{\dG(x|z )}\right)\dP(x,y,z)
\end{align*}
Note $\mathbb{P}_{\rx|y,z}\ll\mathbb{P}_{\rx|z}\ll\mathbb{G}_{\rx|z}\sim\mathbb{G}_{\rx|y,z}, \mathbb{P}_{\ry,\rz} \text{ a.s.}$. By the Radon-Nikodym chain rule, $\KL_{\mathbb{P},\mathbb{G}}(\rx|\ry,\rz) $ follows:
\begin{align*}
    &\KL_{\mathbb{P},\mathbb{G}}(\rx|\ry,\rz)\\ 
    =& \int \log\left(\frac{\dP(x|y,z)}{\dG(x|y,z)}\right)\dP(x,y,z)\\
    =&\int \log\left(\frac{\dP(x|y,z)}{\dP(x|z)}\cdot\frac{\dP(x|z)}{\dG(x|z)}\cdot\frac{\dG(x|z)}{\dG(x|y,z)}\right)\dP(x,y,z)\\
    =&I_{\mathbb{P}}(\rx;\ry|\rz) + \KL_{{\mathbb{P},\mathbb{G}}}(\rx|\rz) -I_{\mathbb{G}}(\rx;\ry|\rz)\\
    < &\infty
\end{align*}
 
\end{proof}

\begin{lemma}\label{lemma:IP-IG}
Let $\mathbb{P}(y^{[k]}_{[T+1]},u_{[T]})$ be the measure generated by a feasible policy $\gamma\in \Gamma^{*}$, \eqref{eqn:control-equation} and $\eqref{eqn:linear-sensor}$. Let $\mathbb{G}(y^{[k]}_{[T+1]},u_{[T]})$ be the Gaussian measure with the same first and second moments as $\mathbb{P}$. Assume that $I_{\mathbb{P}}(\ry^{[i,k-1]}_{[t]},\ry^{k}_{t};\ru_{t}|\ru_{[t-1]},\ry^{[i-1]}_{[t]})$ is finite for any $t\in[T]$. Under the condition \ref{con:fullobservation}, we have

\noindent (a)  For $t\in[T]$, $\KL_{\mathbb{P},\mathbb{G}}(\ry^{[i,k-1]}_{[t]},\ry^{k}_t|\ru_{[t-1]},\ry^{[i-1]}_{[t]})$ is finite;

\noindent (b) For $t\in[T]$, the following holds 
\begin{align*}
    \mathbb{P}_{\ry^{[i,k-1]}_{[t]},\ry^{k}_t,\ry^{[k]}_{t+1}|u_{[t]},y^{[i-1]}_{[t]}}\ll\mathbb{G}_{\ry^{[i,k-1]}_{[t]},\ry^{k}_t,\ry^{[k]}_{t+1}|u_{[t]},y^{[i-1]}_{[t]}}
\end{align*}
$\mathbb{P}_{\ru_{[t]},\ry_{[t]}^{[i-1]}}$ a.s., and $\KL_{\mathbb{P},\mathbb{G}}(\ry^{[i,k-1]}_{[t]},\ry^{k}_t|\ru_{[t]},\ry^{[i-1]}_{[t]})$ is finite and follows the following decomposition:
\begin{align*}
    &\KL_{\mathbb{P},\mathbb{G}}(\ry^{[i,k-1]}_{[t]},\ry^{k}_t|\ru_{[t]},\ry^{[i-1]}_{[t]})\\
    =&\KL_{\mathbb{P},\mathbb{G}}(\ry^{[i-1]}_{t+1}|\ru_{[t]},\ry^{[i-1]}_{[t]}) + \KL_{\mathbb{P},\mathbb{G}}(\ry^{[i,k-1]}_{[t+1]},\ry^{k}_{t+1}|\ru_{[t]},\ry^{[i-1]}_{[t+1]})\\
     &+ \KL_{\mathbb{P},\mathbb{G}}(\ry^{k}_t|\ru_{[t]},\ry^{[k-1]}_{[t+1]},\ry^{k}_{t+1}).
\end{align*}

\end{lemma}

\begin{proof}
We show this by induction. For $t=1$, 
\begin{align*}
    \KL_{\mathbb{P},\mathbb{G}}(\ry^{[i,k-1]}_{1},\ry^{k}_1|\ry^{[i-1]}_{1}) = 0
\end{align*}
since we assume the initial distribution is Gaussian. By Lemma \ref{lemma:one-step}, $\KL_{\mathbb{P},\mathbb{G}}(\ry^{[i,k-1]}_{1},\ry^{k}_1|\ru_{1},\ry^{[i-1]}_{1})$ is finite. Hence, 
\begin{align*}
     \mathbb{P}_{\ry^{[i,k-1]}_{1},\ry^{k}_1|u_{1},y^{[i-1]}_{1}}\ll\mathbb{G}_{\ry^{[i,k-1]}_{1},\ry^{k}_1|u_{1},y^{[i-1]}_{1}}
\end{align*}
holds $\mathbb{P}_{\ru_1,\ry_1^{[i-1]}}$ a.s.. By Lemma \ref{lemma:P-equal-G}, $\mathbb{G}(y^{[k]}_{2}|u_1,y^{[k]}_1) = \mathbb{P}(y^{[k]}_{2}|u_1,y^{[k]}_1)$ for any $(u_1,y^{[k]}_1)$. Hence, multiplying this common forward kernel yields that
\begin{align*}
    \mathbb{P}_{\ry^{[i,k-1]}_{1},\ry^{k}_1,\ry^{[k]}_{2}|u_{1},y^{[i-1]}_{1}}\ll\mathbb{G}_{\ry^{[i,k-1]}_{1},\ry^{k}_1,\ry^{[k]}_{2}|u_{1},y^{[i-1]}_{1}}
\end{align*}
 holds $\mathbb{P}_{\ru_1,\ry_1^{[i-1]}}$ a.s., and we can write
\begin{align*}
    &\KL_{\mathbb{P},\mathbb{G}}(\ry^{[i,k-1]}_{1},\ry^{k}_1|\ru_{1},\ry^{[i-1]}_{1}) \\
    =&\int\log\left(\frac{\dP(y^{[i,k-1]}_{1},y^{k}_1|u_{1},y^{[i-1]}_{1})}{\dG(y^{[i,k-1]}_{1},y^{k}_1|u_{1},y^{[i-1]}_{1})}\right)\dP(y_1^{[k]},u_1)\\
    =&\int\log\left(\frac{\dP(y^{[i,k-1]}_{1},y^{k}_1,y_2^{[k]}|u_{1},y^{[i-1]}_{1})}{\dG(y^{[i,k-1]}_{1},y^{k}_1,y_2^{[k]}|u_{1},y^{[i-1]}_{1})}\right)\dP(y_{[2]}^{[k]},u_1)\\
    =&\int\log\left(\frac{\dP(y^{[i-1]}_{2}|u_{1},y^{[i-1]}_{1})}{\dG(y^{[i-1]}_{2}|u_{1},y^{[i-1]}_{1})}\right)\dP(y_{[2]}^{[k]},u_1)\\
    &+\int\log\left(\frac{\dP(y^{[i,k-1]}_{[2]},y^{k}_{2}|u_{1},y^{[i-1]}_{[2]})}{\dG(y^{[i,k-1]}_{[2]},y^{k}_{2}|u_{1},y^{[i-1]}_{[2]})}\right)\dP(y_{[2]}^{[k]},u_1)\\
    &+\int\log\left(\frac{\dP(y^{k}_{1}|u_{1},y^{[k-1]}_{[2]},y^{k}_2)}{\dG(y^{k}_{1}|u_{1},y^{[k-1]}_{[2]},y^{k}_2)}\right)\dP(y_{[2]}^{[k]},u_1)\\
    =&\KL_{\mathbb{P},\mathbb{G}}(\ry_2^{[i-1]}|\ru_1,\ry_1^{[i-1]}) + \KL_{\mathbb{P},\mathbb{G}}(\ry_{[2]}^{[i,k-1]},\ry_2^{k}|\ru_1,\ry_{[2]}^{[i-1]})\\
    &+\KL_{\mathbb{P},\mathbb{G}}(\ry_1^{k}|\ru_1,\ry_{[2]}^{[k-1]},\ry^{k}_2)
\end{align*}
where the third equality is due to the joint absolute continuity. Therefore, the Radon-Nikodym chain rule applies.

So both Lemma \ref{lemma:IP-IG}(a) and (b) hold for $t=1$. Suppose they hold for $t<T$. We show they also hold for $t+1$. Since Lemma \ref{lemma:IP-IG}(b) holds for $t$ and the L.H.S. of the decomposition is finite, $\KL_{\mathbb{P},\mathbb{G}}(\ry^{[i,k-1]}_{[t+1]},\ry^{k}_{t+1}|\ru_{[t]},\ry^{[i-1]}_{[t+1]})$ is finite. This proves \ref{lemma:IP-IG}(a) for $t+1$. Following the same arguments at $t=1$ using Lemma \ref{lemma:one-step}, one can verify that \ref{lemma:IP-IG}(b) also holds for $t+1$. This completes the induction.   
\end{proof}

\bibliographystyle{ieeetr}
\bibliography{IEEEabrv,Refs}

\begin{IEEEbiographynophoto}{Sijie Li} (Graduate Student Member, IEEE) received the B.S. degree in Mathematics and Information Engineering from the Chinese University of Hong Kong, Hong Kong, in 2022. He is currently a PhD student at the Electrical and Computer Engineering Department of the University of Texas at Austin, affiliated with the Wireless Networking and Communications Group.
\end{IEEEbiographynophoto}

\begin{IEEEbiography}[{\includegraphics[width=1in,height=1.25in,clip,keepaspectratio]{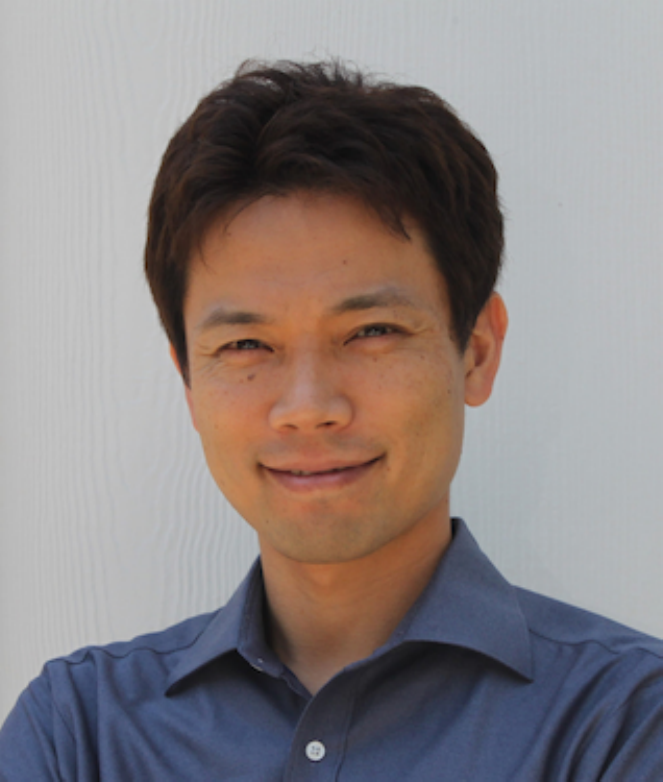}}]
{Takashi Tanaka} (Senior Member, IEEE) received the B.S. degree in Aerospace Engineering from the University of Tokyo, Tokyo, Japan, in 2006, and the M.S. and Ph.D. degrees in Aerospace Engineering (Automatic Control) from the University of Illinois at Urbana-Champaign, Urbana, IL, USA, in 2009 and 2012, respectively.
He was a Postdoctoral Associate with the Laboratory for Information and Decision Systems, Massachusetts Institute of Technology, Cambridge, MA, USA, from 2012 to 2015, and a Postdoctoral Researcher at KTH Royal Institute of Technology, Stockholm, Sweden, from 2015 to 2017.
He was an Assistant Professor in the Department of Aerospace Engineering and Engineering Mechanics at the University of Texas at Austin, Austin, TX, USA between 2017 and 2024, where he was an Associate Professor in 2024.
Since 2025, he has been an Associate Professor of Aeronautics and Astronautics and of Electrical and Computer Engineering at Purdue University, West Lafayette, IN, USA.
Dr. Tanaka was the recipient of the DARPA Young Faculty Award, the AFOSR Young Investigator Program Award, and the NSF Career Award.
\end{IEEEbiography}

\begin{IEEEbiography}[{\includegraphics[width=1in,height=1.25in,clip,keepaspectratio]{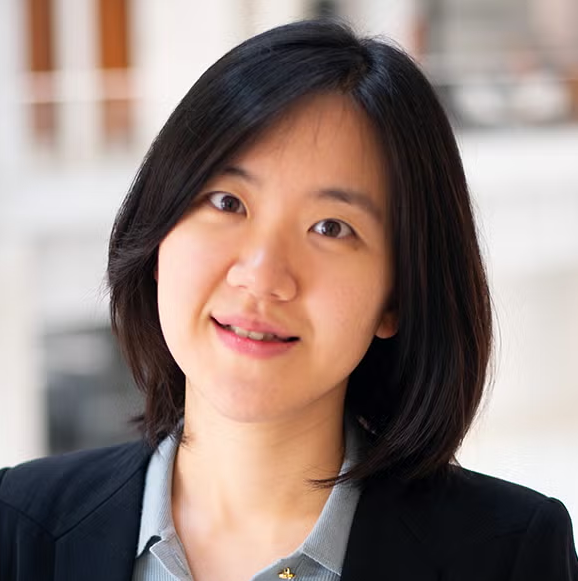}}]{Hyeji Kim} (Senior Member, IEEE) is an Assistant Professor in the Chandra Department of Electrical and Computer Engineering at the University of Texas at Austin. She received her Ph.D. degree in Electrical Engineering from Stanford University in 2016. Before joining UT Austin in 2020, she was a Postdoctoral Researcher at the University of Illinois at Urbana-
Champaign (2016–2018) and a Researcher at Samsung AI Research Cambridge (2018–2020). At UT Austin, she is a member of the Wireless Networking and Communications Group (WNCG) and an affiliate of the Machine Learning Laboratory (MLL). Her research lies at the intersection of information theory, machine learning, and wireless communications. She received the NSF CAREER Award in 2024.
\end{IEEEbiography}

\end{document}